\documentclass[twoside,leqno,twocolumn]{article}
\usepackage{ltexpprt}

\usepackage{amsmath,amssymb}
\usepackage{paralist}

\usepackage{tikz}
\usetikzlibrary{positioning}

\newcommand{\ignore}[1]{}
\usepackage{color}

\newcommand{\tU}{\tilde{U}}
\newcommand{\sign}{\mathrm{sign}}

\newcommand{\abs}[1]{| #1 |}
\newcommand{\norm}[1]{\| #1 \|}

\newcommand{\set}[2]{\left\{ #1 \;|\; #2 \right\}}
\newcommand{\sset}[1]{\left\{ #1 \right\}}
\newcommand{\xmax}{x_{\max}}
\newcommand{\xeq}{x_{\mathit{eq}}}

\newcommand{\ceil}[1]{\lceil #1 \rceil}

\newcommand{\N}{\mathbb {N}}
\newcommand{\tu}{\tilde{u}}
 \newcommand{\tp}{\tilde{p}}

\newcommand{\hp}{\hat{p}}
\newcommand{\hG}{\hat{G}}
\newcommand{\hf}{\hat{f}}
\newcommand{\hr}{\hat{r}}
\newcommand{\hx}{\hat{x}}

\newcommand{\outflow}{\mathit{outflow}}
\newcommand{\inflow}{\mathit{inflow}}
\newcommand{\Enew}{E_{\mathit{new}}}

\newcommand{\barS}{\bar{S}}
\newcommand{\rmin}{r_{\min}} 
\newcommand{\rmax}{r_{\max}}
\newcommand{\define}{\mathrel{:=}}

\newcommand{\tOmega}{\widetilde{\Omega}}

\newcommand{\DM}{Duan-Mehlhorn:Arrow-Debreu-Market}

\begin{document}

\title{An Improved Combinatorial Polynomial Algorithm for the Linear Arrow-Debreu Market}
%\tnotetext[t1]{A preliminary version of this paper was presented at ICALP 2013.}

\author{Ran Duan\thanks{Institute for Interdisciplinary Information Sciences, Tsinghua University, Beijing, China}
\and Jugal Garg\thanks{Max-Planck-Institut f\"ur Informatik, Saarbr\"ucken, Germany}
 \and Kurt Mehlhorn\thanks{Max-Planck-Institut f\"ur Informatik, Saarbr\"ucken, Germany}
}

\date{}

% \author[RD]{Ran Duan}
% \ead{duanran@mail.tsinghua.edu.cn}
% \author[JG]{Jugal Garg}
% \ead{jugal@mpi-inf.mpg.de}
% \author[KM]{Kurt Mehlhorn}
% \ead{mehlhorn@mpi-inf.mpg.de}

\maketitle

%\pagenumbering{arabic}
%\setcounter{page}{1}%Leave this line commented out.

%\tableofcontents

\begin{abstract} We present an improved combinatorial algorithm for the computation of equilibrium prices in the linear Arrow-Debreu model. For a market with $n$ agents and integral utilities bounded by $U$, the algorithm runs in $O(n^7 \log^3 (nU))$ time. This improves upon the previously best algorithm of Ye by a factor of $\tOmega(n)$. The algorithm refines the algorithm described by Duan and Mehlhorn and improves it by a factor of $\tOmega(n^3)$. The improvement comes from a better understanding of the iterative price adjustment process, the improved balanced flow computation for nondegenerate instances, and a novel perturbation technique for achieving nondegeneracy. \end{abstract}

\section{Introduction}\label{sec: introduction}

Walras~\cite{Walras1874} introduced an economic market model in 1874. In this model, every agent has an initial
endowment of some goods and a utility function over sets of goods. Goods are assumed to be divisible. 
The market clears at a set of prices if each agent can spend its entire budget (= the total value of its goods at the set of prices) on a bundle of goods with maximum utility, and all goods are completely sold. Market clearing prices are also called equilibrium prices. In the linear version of the model, all utility functions are linear. 

We present an improved combinatorial algorithm for the computation of equilibrium prices. For a market with $n$ agents and integral utilities bounded by $U$, the algorithm runs in $O(n^7 \log^3 (nU))$ time. This is almost by a factor of $n$ better than all known algorithms. Jain~\cite{Jain07} and Ye~\cite{Ye2007} gave algorithms based on the ellipsoid and the interior point method, respectively, and Duan and Mehlhorn~\cite{\DM} described a combinatorial algorithm. The new algorithm refines the latter algorithm and improves upon its running time by a factor of almost $n^3$.

Formally, the linear model is defined as follows. We may assume w.l.o.g.~that the number of goods is equal to the number of agents and that the $i$-th agent owns the $i$-th good. There is one unit of each good. Let $u_{ij} \ge 0$ be the utility for agent $i$ if all of good $j$ is allocated to him. We assume that the utilities are integers bounded by $U$. 
%If a fraction $x_{ij}$ of good $j$ is allocated to buyer $i$, the total utility for  $i$ is $ \sum_{j} u_{ij} x_{ij}$.
We make the standard assumption that each agent likes some good, i.e., for all $i$, $\max_j u_{ij} > 0$, and that each good is liked by some agent, i.e., for all $j$, $\max_i u_{ij} > 0$.  We also make the nonstandard assumption that for every proper subset $P$ of agents, there exist $i\in P$ and $j\notin P$ such that $u_{ij} > 0$. References~\cite{Jain07,Devanur-Garg-Vegh,Duan-Mehlhorn:Arrow-Debreu-Market} show how to remove it.
Agents are rational and spend money only on goods that give them maximum utility per unit of money, i.e., if $p = (p_1,\ldots,p_n)$ is a price vector then buyer $i$ is only willing to buy goods $j$ with $u_{ij}/p_j = \max_\ell
  u_{i\ell}/p_\ell$. An \emph{equilibrium} is a vector $p$ of positive prices and allocations $x_{ij} \ge 0$ such
  that
\begin{itemize}
\item all goods are completely sold:\quad $\sum_i x_{ij} = 1$ for all $j$.
\item all money is spent: \quad $p_i = \sum_j x_{ij} p_j$ for all $i$.
\item only goods that give maximum utility per unit of money are bought:
\[    \text{for all $i$ and $j$:} \quad x_{ij} > 0 \quad \Rightarrow \quad \frac{u_{ij}}{p_j} = \alpha_i,\]
where $\alpha_i = \max_\ell \frac{u_{i\ell}}{p_\ell}$. 
\end{itemize}
In an equilibrium, $f_{ij} = x_{ij} p_j$ is the amount of money that flows from agent $i$ to agent $j$. In this paper, it is useful to represent each agent $i$ twice, once in its role as a buyer and once in its role as (the owner of) a good. We denote the set of  buyers by
$B=\{b_1,b_2,...,b_n\}$ and the set of goods by
$C=\{c_1,c_2,...,c_n\}$. So, if the price of good $c_i$ is $p_i$,
buyer $b_i$ will have $p_i$ amount of money.

The existence of an equilibrium is nonobvious. The first rigorous existence proof is due to Wald~\cite{Wald36}.
It required fairly strong assumptions. The existence proof for the
general model was given by Arrow and Debreu~\cite{AD1954} in 1954.
They proved that the market equilibrium always exists if
the utility functions are concave. The proof is nonconstructive. Gale~\cite{Gale57,Gale76} gave necessary and sufficient conditions for the existence of an equilibrium in the linear model.

The development of algorithms started in the 60s. We restrict the discussion to exact algorithms and refer the  reader to~\cite{\DM} for a broader discussion of algorithms and related work.
Eaves~\cite{Eaves76} presented the first exact algorithm for the linear exchange model. He reduced the model to a linear complementary problem which is then solved by Lemke's algorithm. The algorithm is not polynomial time. Garg, Mehta, Sohoni, and Vazirani~\cite{Garg-Mehta-Sohoni-Vazirani} give a combinatorial interpretation of the algorithm. Polynomial time exact algorithms were obtained based on the 
characterization of equilibria as the solution set of a convex program. The recent paper by Devanur, Garg, and V\'{e}gh~\cite{Devanur-Garg-Vegh} surveys these characterizations. 
% For example, Nenakov and Primak~\cite{Nenakov-Primak} showed that equilibria are precisely the solutions of the system
% \[  p_i > 0 \quad x_{ij} \ge 0\quad \sum_j u_{ij} x_{ij} \ge \frac{u_{ik}}{p_k} p_i  \quad\text{for all $i$ and $k$}.\]
% Note that $\frac{u_{ik}}{p_k} p_i$ is the utility obtained by the $i$-th buyer if he spends his entire budget, which is $p_i$ under the assumption that his good is completely sold, on good $k$. The program above is not convex in the variables $p_i$ and $x_{ij}$. However, after the substitution $p_i = e^{\pi_i}$ it becomes a convex program in the variables $\pi_i$ and $x_{ij}$. 
Jain~\cite{Jain07} showed how to solve one of these characterizations with a nontrivial extension of the ellipsoid method. His algorithm is the first polynomial time algorithm for the linear Arrow-Debreu market. Ye~\cite{Ye2007} showed that polynomial time can also be achieved with the interior point method. We quote from his paper: ``\emph{We present an interior-point algorithm \ldots\ for
solving the Arrow{-}Debreu pure exchange market equilibrium problem with linear utility. \ldots\ The number
of arithmetic operations of the algorithm is \ldots\ bounded by $O(n^4 \log 1/\epsilon)$, which is substantially lower than the one obtained by the ellipsoid method. If
the input data are rational, then an exact solution can be obtained by solving the
identified system of linear equations and inequalities, 
when $\epsilon < 2^{-L}$, where $L$ is the bit length of the input data. Thus, the arithmetic
operation bound becomes $O(n^4 L)$, \ldots.}'' We assume that utilities are integers bounded by $U$. Then $L = n^2 \log U$, and the number of arithmetic operations becomes $O(n^6 \log U)$. Ye does not state the precision needed for the computation. However, since the computation must guarantee that the error $\epsilon$ becomes less than $2^{-L}$, it is fair to assume that arithmetic on numbers with $L$ bits is necessary. Thus the time complexity on a RAM becomes $O(n^6 (\log U) \cdot M(n^2 \log U))$, where $M(L)$ is the cost of multiplying $L$-bit numbers. On a RAM with logarithmic word-length, $M(L) = O(L)$~\cite{Furer09} and hence the time bound for Ye's algorithm\footnote{Yinyu Ye confirmed in personal communication that this interpretation of his result is correct.}  becomes $O(n^8 (\log U)^2)$.

The utility graph is a bipartite graph with vertex set $B \cup C$, where $b_i$ and $c_j$ are connected by an edge if and only if $u_{ij} > 0$. Any cycle $D$ in the utility graph has even length and hence the edges of the cycle can be 2-colored such that adjacent edges have distinct colors. We use $D_0$ and $D_1$ to denote the two color classes and call $(D_0,D_1)$ the 2-partition of $D$. We call an instance \emph{degenerate} if there is a cycle $D$ with 2-partition $(D_0,D_1)$ such that 
\begin{equation}\label{degeneracy} \prod_{e \in D_0} u_e = \prod_{e \in D_1} u_e. \end{equation}
For a price vector $p$, the equality graph $G_p = (B \cup C,E_p)$ is a directed bipartite graph, where the edge set $E_p$ consists of all edges $(b_i,c_j)$ such that $u_{ij}/p_j = \alpha_i$. 
For nondegenerate instances, the equality graph with respect to any price vector is a forest, see Lemma~\ref{general position}. 

Duan and Mehlhorn~\cite{\DM} gave a combinatorial algorithm. They
obtain equilibrium prices by a procedure that iteratively adjusts prices and allocations in a carefully chosen, but intuitive manner. The algorithm runs in $O(n^5 \log (nU))$ phases and maintains a balanced flow in a flow network defined by the current price vector. A balanced flow is a maximum flow that minimizes the 2-norm of the surplus vector of the agents. The balanced flow needs to be recomputed in each phase and this takes $O(n)$ max-flow computations on a graph with $n$ nodes and up to $n^2$ edges. We refine their algorithm (see Figure~\ref{fig:algo} for a
complete listing) and reduce the running time by almost three orders of magnitude. The improvement comes from three sources. 
\begin{itemize}
\item A refined, but equally simple, method for adjusting prices and a refined analysis. This reduces the number of phases by a factor of $n$. 
\item An improved computation of balanced flows in nondegenerate networks. We show that only one general max-flow computation in a graph with $O(n)$ edges is needed; the other $O(n)$ maxflow computations are in networks with forest structure. This reduces the number of arithmetic operations per phase by a factor of $n^2$
and improves upon the running time by a factor of $n^2/\log(nU)$. The improvement also applies to the algorithm in~\cite{\DM}. 
\item A novel perturbation scheme that removes degeneracies without increasing the cost of our algorithm. 
\end{itemize}
Orlin~\cite{Orlin10} previously described a perturbation scheme. We will argue in the appendix that his description is incomplete. The Cole-Gkatzelis algorithm~\cite{Cole-Gkatzelis} for Nash social welfare makes use of Orlin's perturbation scheme. Our perturbation scheme is applicable in both algorithms. 

This paper is organized as follows. We introduce some notation and concepts in Section~\ref{sec: Preliminaries} and describe the algorithm in Section~\ref{sec: algorithm}. Sections~\ref{sec: Analysis} and~\ref{sec: Evolution} prove the bound on the number of phases. Section~\ref{sec: Balanced Flows} contains the improved algorithm for computing balanced flows for nondegenerate instances and Section~\ref{sec: Perturbation} introduces the novel perturbation scheme.

% \subsection{Further Related Work}\label{sec:relatedwork}

% There are also algorithms for the Arrow-Debreu model with
% non-linear utilities~\cite{CMV05,CMPV05}. The CES (constant
% elasticity of substitution) utility functions have drawn much
% attention; here, the utility functions are of the form
% $u(x_1,...,x_n)=(\sum_{j=1}^nc_jx_j^\rho)^{1/\rho}$ for $-\infty<\rho<1$
% and $\rho\neq0$. Codenotti, McCune, Penumatcha, and Varadarajan~\cite{CMPV05} have shown that for $\rho>0$ and $-1\leq\rho<0$, there are polynomial algorithms based on a convex program. In contrast,
% Chen, Paparas and Yannakakis~\cite{CPY12} have shown that it is
% PPAD-hard to solve market equilibrium of CES utilities for $\rho<-1$.
% They also define a new concept ``non-monotone utilities'', and
% show the PPAD-hardness to solve the markets with non-monotone
% utilities. It remains open to find the exact border between
% tractable and intractable utility functions.

% All algorithms mentioned so far are centralized in the sense that the algorithms need to know all the utilities at the start of the algorithm and, in the case of iterative algorithms, the global state of the market, i.e., all prices and the demand for each good. Cole and Fleischer~\cite{Cole-Fleischer} and Cheung, Cole, and Rastogi~\cite{Cheung-Cole-Rastogi} explore local algorithms, where each agent only knows the price and demand of his good and, moreover, the market is run for many periods. 

\section{Preliminaries}\label{sec: Preliminaries}

% The linear exchange model is defined by the following assumptions; they are as in Jain's paper~\cite{Jain07}:\medskip

% \begin{compactenum}[1.]
%     \item There are $n \ge 2$ agents in the system. Each agent $i$ has only one good, which is different from the goods other people have. Agent $i$ owns good $i$.
%     \item There is one unit of each good $i$. So, if the price of good $i$ is $p_i$, agent $i$ will obtain $p_i$ units of money when selling its good completely. 
%     \item Each agent $i$ has a linear utility function $\sum_j u_{ij}x_{ij}$, where $x_{ij}$ is the amount of good $j$ consumed by $i$.
%     \item For all $i$, there is a $j$ such that $u_{ij}>0$. (Everybody likes some good.)
%     \item For all $j$, there is an $i$ such that $u_{ij}>0$. (Every good is liked by somebody.)
%     \item (Irreducibility) For every proper subset $P$ of agents, there exist
%     $i\in P$ and $j\notin P$ such that $u_{ij}>0$.
%   \item Each $u_{ij}$ is an integer between 0 and $U$.
% \end{compactenum}\medskip

% Assumptions 1 to 5 are standard. Assumption 6 simplifies the presentation and  references~\cite{Jain07,Devanur-Garg-Vegh,Duan-Mehlhorn:Arrow-Debreu-Market} show how to remove it. Assumption 7 is for the purposes of polynomial time computation. 

%For a subset $S$ of buyers, we use $p(S)$ to denote the total prices of the goods owned by them. 
For a vector $v =(v_1,\ldots,v_n)$, let $\abs{v} = \abs{v_1} + \ldots + \abs{v_n}$ and $\norm{v} = (v_1^2 + \ldots + v_n^2)^{1/2}$ be the $\ell_1$ and $\ell_2$-norm of $v$, respectively. 

Let $p=(p_1,p_2,...,p_n)$ denote the vector of prices of goods $1$
to $n$, so they are also the budgets of agents $1$ to $n$, if goods are completely sold. Each agent only buys its favorite goods, that is, the
goods with the maximum ratio of utility and price. Define the \emph{bang per buck} of buyer $b_i$ to be
$\alpha_i=\max_j\{u_{ij}/p_j\}$. 
For a price vector $p$, the \emph{equality network} $N_p$ is a flow network with vertex set
$\sset{s,t} \cup B \cup C$, where $s$
is a source node, $t$ is a sink node, $B$ is the set of buyers, and $C$ is the set of goods, and the following edge set: \medskip

\begin{compactenum}[\hspace{\parindent}(1)]
    \item An edge $(s,b_i)$ with capacity $p_i$ for each $b_i \in B$. 
    \item An edge $(c_i,t)$ with capacity $p_i$ for each $c_i \in C$.
    \item An edge $(b_i,c_j)$ with infinite capacity whenever $u_{ij}/p_j=\alpha_i$. We use $E_p$ to denote these edges. 
\end{compactenum}\medskip

Our task is to find a positive price vector $p$ such that there is a
flow in which all edges from $s$ and to $t$ are saturated, i.e.,
$(s,B\cup C\cup t)$ and $(s\cup B\cup C, t)$ are both minimum
cuts. When this is satisfied, all goods are sold and all of the
money earned by each agent is spent on goods of maximum utility per unit of money.

For a set $S$ of buyers define its neighborhood $\Gamma(S)=\set{c \in C}{\text{$(b,c)\in
E_p$ for some $b \in S$}}$. 
Clearly, there is no edge in $E_p$ from $S$ to
$C\setminus\Gamma(S)$.

With respect to a flow $f$, define the surplus $r(b_i)$ of a buyer $i$ as $r(b_i)=p_i-\sum_j f_{ij}$,  where $f_{ij}$ is the amount of flow on  the edge $(b_i,c_j)$, and define the surplus $r(c_j)$ of a good $j$ as $r(c_j)=p_j-\sum_i f_{ij}$, Define the surplus vector of buyers to be
$r =(r(b_1),r(b_2),...,r(b_n))$. Also, define the total surplus
to be $|r|=\sum_i r(b_i)$, which is also $\sum_j r(c_j)$ since
the total capacity from $s$ and to $t$ are both equal to $\sum_i
p_i$. For convenience, we denote the surplus vector of flow $f'$
by $r'$. In the network corresponding to market clearing
prices, the total surplus of a maximum flow is zero. For a set $X$ of buyers, let $\rmin(X) = \min \set{r(b)}{b \in X}$ and $\rmax(X) = \max \set{r(b)}{b \in X}$ be the minimal and maximal surplus of any buyer in $X$. For the empty set of buyers, $\rmax(\emptyset) = 0$. The outflow of buyer $b_i$ is $\outflow(b_i) = \sum_j f_{ij}$. 

A maximum flow is \emph{balanced} if it minimizes the 2-norm of the surplus vector of the buyers. The concept of a balanced flow was introduced by Devanur et al.~\cite{DPSV08}.

\begin{lemma}[\cite{DPSV08,\DM}]\label{lem:balanced flow} Balanced flows exist. They can be computed with $n$ maximum flow computations. If $f$ is a balanced flow, buyers $b_i$ and $b_j$ have equality edges connecting them to the same good $c$ and there is positive flow from $b_i$ to $c$, then the surplus of $b_j$ is no larger than the surplus of $b_i$. Let $f_0$ be a maximum flow and let $C_0$ be the goods that are completely sold with respect to $f_0$. Then there is a balanced flow in which all goods in $C_0$ are completely sold. \end{lemma}

We next show that for nondegenerate instances the equality graph with respect to any price vector is a forest.

\begin{lemma}\label{general position}
Consider any cycle $D$ in the utility graph. If $D \subseteq E_p$ for some price vector $p$, the set of utilities is degenerate. 
\end{lemma}
\begin{proof}
Consider any buyer $b$ in the cycle and let $e_0 = (c_0,b) \in D_0$ and $e_1 = (b,c_1) \in D_1$ be the edges in $D$ incident to $b$. Then ${u_{e_0}}/{p(c_0)}   = {u_{e_1}}/{p(c_1)}$, and therefore $p(c_1)= ({u_{e_1}}/{u_{e_0}}) p(c_0)$, and hence ${\prod_{e \in D_1} u_e}/{\prod_{e \in D_0} u_e} = 1$. 
\end{proof}

\paragraph{Fisher Markets:} In Fisher markets, the buyers come with a budget to the market. They do not have to earn their budget by selling their goods. Fisher markets are a special case of Arrow-Debreu markets. Let $a_i$ be the budget of the $i$-th buyer. Consider the Arrow-Debreu market, where the $i$-th buyer owns $a_i$ units of each good. Then in an equilibrium his budget will be $a_i \sum_j p_j$ and hence a solution to the Fisher market can be obtained from the solution to the Arrow-Debreu market by dividing all prices and money flows by $P = \sum_j p_j$.

\section{The Algorithm}\label{sec: algorithm}

The algorithm is shown in Figure~\ref{fig:algo} and refines the
one of Duan and Mehlhorn~\cite{\DM}. It starts with all prices $p_i$ equal to one and a balanced flow $f$ in $N_p$. It works in phases (called iterations in~\cite{\DM}). In each phase, we first determine a set $S$ of buyers with surplus and the set $\Gamma(S)$ of goods connected to them by equality edges. We increase the prices of the goods in $\Gamma(S)$ and the money flow into these goods by a common factor $x > 1$. Let $p'$ be the new price vector and let $f'$ be the resulting flow. We turn $f'$ into a balanced flow $f''$ and make sure that all goods that are completely sold with respect to $f$ are also completely sold with respect to $f''$. We set $f$ to $f''$, $p$ to $p'$, 
and repeat. Once the total surplus is less than $\epsilon$, where $\epsilon = 1/(8 n^{4n} U^{3n})$, we exit from the loop and compute the equilibrium prices from the current price vector $p$ and the current flow $f$. The last step is exactly as in~\cite{\DM} and will not be discussed further.

% \begin{figure}[t]
% \centerline{\framebox{\parbox{0.9\textwidth}{
% \begin{tabbing}
% 555\=555\=555\=555\=\kill
% \>Set all prices $p_i$ to one and compute a maximum flow $f$ in $N_p$;\quad\mbox{}\\ \medskip
% \>{\bf Repeat}\\
% \>\>Adjust prices and flow and obtain a price vector $p'$ and a flow $f'$;\\
% \>\>$p = p'$ and $f = f'$;\\
% \>{\bf Until} the total surplus of $f$ is tiny; \\ \medskip
% \>Round the price vector $p$ to get an exact solution;
% \end{tabbing}}}}
% \caption{A high-level view of the algorithm.}\label{fig:high-level-view}
% \end{figure}

\begin{figure*}[t]
\centerline{\framebox{\parbox{5.0in}{
\begin{tabbing}
555\=555\=555\=555\=\kill
\>{\bf Part A:} {Set $\epsilon = {1}/(8n^{4n}U^{3n})$; }\\[0.3em]
\>Set $p_i=1$ for all $i$ and set $f$ to a balanced flow in $N_p$; \\[0.3em]
\>{\bf While} $|r(B)| \ge \epsilon$;\\[0.3em]
\>\>Sort the buyers by their surpluses in decreasing order: $b_1,b_2,...,b_n$; \\[0.3em]
\>\>\parbox{0.8\textwidth}{Find the smallest  $\ell \ge 1$ for which $S = S(r(b_\ell))$ satisfies $\outflow(b_i) = 0$ and $c_i \not\in \Gamma(S)$ for every $b_i$ with $r(b_\ell) > r(b_i) \ge r(b_\ell)/(1 + 1/n)$ and let $\ell=n$ when there is no such $\ell$;}\\[0.3em]
\>\>Let $S = S(r(b_\ell))$; \\[0.3em]
\>\>\parbox{0.8\textwidth}{Determine $\xmax$, $\xeq$, $x_{23}$, $x_{24}$ and $x_2$ and set $x = \min(\xmax,\xeq,x_{23},x_{24},x_2)$;}\\[0.3em]
\>\>\parbox{0.8\textwidth}{Update prices and flow according to (\ref{new prices}) and (\ref{new flow interior}); let $f'$ be the new flow and $p'$ be the new price vector;}\\[0.3em]
\>\>\parbox{0.8\textwidth}{Compute a balanced flow $f''$ in $N_{p'}$ with the property that goods that are completely sold with respect to $f$ are also completely sold with respect to $f''$;}\\[0.3em]
\>\> Set $f$ to $f''$ and $p$ to $p'$;\\[0.3em]
\>{\bf EndWhile} \\[0.3em]
\>{\bf Part B:} Compute equilibrium prices from $f$ and $p$ as in~\cite{\DM};
\end{tabbing}\vspace{-1em}
}}}\caption{The complete algorithm}\label{fig:algo}
%\end{multicols}
\end{figure*}

% a sequence $S_0(r_0)$, $S_1(r_1)$, \ldots of sets of buyers:
% \begin{align*}
% S_0(r_0) &= \set{b_i}{r(b_i) \ge r_0 \text{ and } \outflow(b_i) > 0}\\
% S_{\ell}(r_0) &= S_{\ell - 1}(r_0) \cup \set{b_i}{r(b_i) \ge r_0 \text{ and } c_i \in \Gamma(S_{\ell - 1}(r_0))} &\text{for $\ell \ge 1$}
% \end{align*}
% Since the sequence is nested, it becomes stable at some point. In fact, it becomes stable almost immediately. Let $S(r_0) = \cup_{\ell \ge 0} S_\ell(r_0)$. 

% \begin{lemma} Let $f$ be a balanced flow in $N_p$. Then $S(r_0) = S_1(r_0)$. \end{lemma}
% \begin{proof} For simplicity of notation, we write $S_\ell$ instead of $S_\ell(r_0)$. Assume otherwise, and let $b_i \in S_2 \setminus S_1$. Then $r(b_i) \ge r_0$ and $c_i \in \Gamma(S_1) \setminus \Gamma(S_0)$. Thus there is a $b_j \in S_1 \setminus S_0$ with $(b_j,c_i) \in E_p$. Also $r(b_j) \ge r_0$. Since $b_j$ has no outflow, $c_i$ must be completely sold. Since the flow is balanced, all flow to $c_i$ comes from buyers whose surplus is at least the surplus of $b_j$ and there is at least one such buyer. This buyer belongs to $S_0$. Thus $c_i \in \Gamma(S_0)$ and hence $b_i \in S_1$, a contradiction. \end{proof}

We next detail the phases. We first explain the choice of $S$. It is more refined than in~\cite{\DM} and crucial for the improved bound on the number of phases. For a resource bound $r_0$, let $S(r_0) = \set{b \in B}{r(b) \ge r_0}$ be the set of buyers with surplus at least $r_0$. In our algorithm we choose a particular resource bound $r_0$. We order the buyers in decreasing order of surpluses: $r(b_1) \ge r(b_2) \ge \ldots \ge r(b_{n-1}) \ge r(b_n)$. Let $\ell \ge 1$ be minimal\footnote{\cite{\DM} uses a simpler definition: let $\ell$ be minimal such that $r(b_{\ell + 1}) < r(b_\ell)/(1 + 1/n)$. For this definition, Lemma~\ref{surplus bounds} does not hold.} such that 
%$\outflow(b_\ell) > 0$ and 
$\outflow(b_j) = 0$ and $c_j \not\in \Gamma(S)$ for every $b_j$ with $r(b_\ell) > r(b_j) \ge r(b_\ell)/(1 + 1/n)$, where $S = S(r(b_\ell))$. If no such $\ell$ exists, set $\ell = n$ and $S = S(r(b_n))$. Let $r_0 = r(b_\ell)$. 

\begin{lemma} $r(b_\ell) > 0$. \end{lemma}
\begin{proof} If the loop is entered, $\abs{r(B)} > 0$. Let $k$ be maximal with $r(b_k) > 0$. Then either $k = n$ or $r(b_{k+1}) = 0 < r(b_k)/(1 + 1/n)$. Thus $\ell \le k$. \end{proof}

The index $\ell$ is readily determined. 

\begin{lemma} The index $\ell$ can be found in $O(n^2)$ time. \end{lemma}
\begin{proof} We use the straightforward algorithm of scanning through the buyers in order of decreasing surplus. \smallskip

Set $\ell$ to 1, $\Gamma$ to $\Gamma(b_1)$, and enter a loop.

Increment $\ell$ as long as $\ell < n$ and $r(b_{\ell + 1}) = r(b_\ell)$. Whenever $\ell$ is increased, add $\Gamma(b_\ell)$ to $\Gamma$. Maintain $\Gamma$ as a bit-vector. 

If $\ell$ reaches $n$ return it as the desired value.

Otherwise, do the following for $j = \ell+1$, $\ell+2$, \ldots: if $j$ reaches $n+1$ or $r(b_j) < r(b_\ell)/(1 + 1/n)$, return $\ell$. If $\outflow(b_j) \not= 0$ or $c_j \in \Gamma$, add $\Gamma(\sset{b_{\ell + 1}, \ldots, b_j})$ to $\Gamma$, set $\ell = j$, and break from the inner loop. \smallskip

The algorithm runs in time proportional to the number of edges of the utility graph and hence in time $O(n^2)$. 
\end{proof}

% \begin{lemma} Let $b$ be such that $r(b) \ge \rmin(S)/(1 + 1/n)$ and $\outflow(b) = 0$. Let $c$ be any good such that $(b,c) \in E_p$. Then $c \in \Gamma(S)$. \end{lemma}
% \begin{proof} If $b$ has no outflow, its surplus is equal to the price of the good owned by $b$. Thus $b$ has surplus and hence $c$ must be completely sold. The buyers with flow to $c$ have surplus at least the surplus to $b$. Since they have outflow, they even have surplus at least $\rmin(S)$ and hence belong to $S$. Thus $c \in \Gamma(S)$. 
% \end{proof}

\begin{lemma}\label{price increase and network} There are no edges in $E_p$ from $S$ to $C \setminus \Gamma(S)$, and the edges from $\barS$ to $\Gamma(S)$ are not carrying flow.
\end{lemma} 
\begin{proof} The first claim is immediate from the definition of $\Gamma(S)$. For the second claim, assume otherwise and let $b_j \not\in S$  and $c_k \in \Gamma(S)$ be such that there is positive flow on the edge $(b_j,c_k)$. Since $c_k \in \Gamma(S)$, there must a buyer $b_i \in S$ with $(b_i,c_k) \in E_p$. Since the flow is balanced, $r(b_j) \ge r(b_i)$ by Lemma~\ref{lem:balanced flow}, and hence $b_j \in S$. \end{proof}

We raise the prices of the goods in $\Gamma(S)$ and the flow on the edges incident to them 
by a  common factor $x > 1$. We also increase the flow from $s$ to buyers in $S$ such that flow conservation holds. 
This give us a new price vector $p'$ and a new flow $f'$. Except for the modified definition of $S$, this is exactly as in~\cite{Duan-Mehlhorn:Arrow-Debreu-Market}. Observe that the surpluses of the goods in $\Gamma(S)$ stay zero.
Formally, 
\begin{align}\label{new prices}
p'_j &= \begin{cases}
x\cdot p_j     & \text{if $c_j\in \Gamma(S)$};\\
p_j & \text{if $c_j\notin\Gamma(S)$}.
\end{cases}\\
\label{new flow interior}
f'_{ij} &= \begin{cases}
x\cdot f_{ij}      & \text{if $c_j\in \Gamma(S)$};\\
f_{ij}  & \text{if $c_j\notin\Gamma(S)$}.
\end{cases}
\end{align}
The changes on the edges incident to $s$ and $t$ are implied by flow conservation.
Since there are no edges from $S$ to $C \setminus \Gamma(S)$, and the edges from $\barS$ to $\Gamma(S)$ are not carrying flow, an equivalent definition of the updated flow is 
$f'_{ij} = x f_{ij}$ if $b_i \in S$ and $f'_{ij} = f_{ij}$ if $b_i \in \barS$.

The change of prices and flows affects the surpluses of the buyers, some go up and some go down. We distinguish between four types of buyers, depending on whether a buyer $b$ belongs to $S$ or not and whether the good
owned by $b$ belongs to $\Gamma(S)$ or not. 
The effect of the change on the surpluses of the buyers is given by the following theorem.

\begin{lemma}[\cite{\DM}]\label{thm:adjusted-flow}
Given a maximum flow $f$ in $N_p$, a set $S$ of buyers such that all goods in $\Gamma(S)$ are completely sold and there is no flow from $\barS$ to $\Gamma(S)$, and a sufficiently small parameter $x > 1$, the flow $f'$ defined in (\ref{new flow interior}) is a feasible flow in the equality network with respect to the prices in (\ref{new prices}). The surplus of each good remains unchanged, and the
surpluses of the buyers become:
\[
r'(b_i)=\left\{
\begin{array}{lr}
x\cdot r(b_i)  & \text{if $b_i\in S, c_i\in\Gamma(S)$} \\
& \text{(type 1 buyer)};  \\
(1-x)p_i+x\cdot r(b_i)     & \text{if $b_i\in S, c_i\notin\Gamma(S)$}\\
& \text{ (type 2 buyer)}; \\
(x-1)p_i+r(b_i)        & \text{if $b_i\notin S, c_i\in\Gamma(S)$}\\
&\text{(type 3 buyer)};\\
r(b_i)          & \text{if $b_i\notin S,
c_i\notin\Gamma(S)$}\\
& \text{(type 4 buyer)}. 
\end{array}
\right.
\]
\end{lemma}
\begin{proof} See~\cite{\DM}. \end{proof}

\begin{figure}[t]
\centerline{\includegraphics[width=0.4\textwidth]{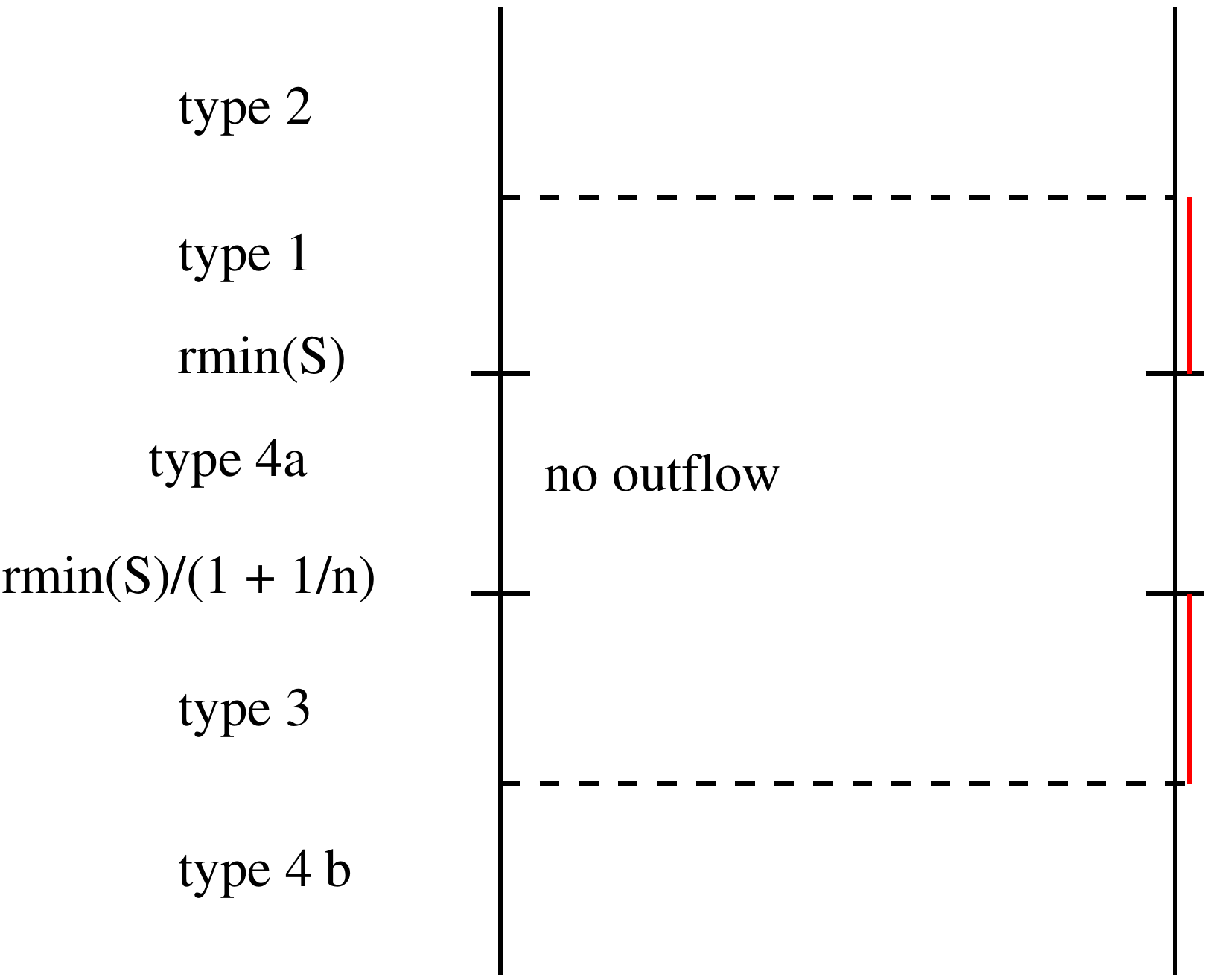}}
\caption{\label{types} The left vertical segment shows the buyers grouped into the buyers with surplus at least $\rmin(S)$, with surplus in $(\rmin(S),\rmin(S)/(1 + 1/n)]$, and with surplus below $\rmin(S)/(1 + 1/n)$. $S$ constitutes the type 1 and 2  buyers. The red part on the right vertical segment indicates $\Gamma(S)$. The goods in $\Gamma(S)$ are completely sold to the buyers in $S$. The type 4a buyers have no outflow and all equality edges incident to them end in $\Gamma(S)$. The type 3 and 4 buyers have no flow to $\Gamma(S)$. }
\end{figure}

In contrast to \cite{\DM}, we need to split the set of type 4 buyers further into type 4a buyers and type 4b buyers. 
A type 4 buyer has type 4a, if its surplus is at least $\rmin(S)/(1 + 1/n)$. Otherwise, its type is 4b. See Figure~\ref{types} for an illustration.

\begin{lemma} $r(b_i) < \rmin(S)/(1 + 1/n)$ for every type 3 and type 4b buyer $b_i$. For a type 4a buyer $b_i$ and an edge $(b_i,c_j) \in E_p$, $c_j \in \Gamma(S)$. \end{lemma}
\begin{proof} Consider any buyer $b_i$ with $r(b_i) \ge \rmin(S)/(1 + 1/n)$ and $b_i \not\in S$. Then $c_i \not\in \Gamma(S)$ and hence $b_i$ does not have type 3. So, it must have type 4, and more precisely, type 4, it has type 4a. Let $(b_i,c_j) \in E_p$. Since $\outflow(b_i) = 0$, $b_i$ has surplus and hence $c_j$ is completely sold. The buyers with flow to $c_j$ have surplus at least the surplus of $b_i$ and hence have surplus at least $\rmin(S)$. Thus they belong to $S$ and hence $c_j \in \Gamma(S)$. 
\end{proof}

\begin{lemma} If there is no type 3 buyer then there is at least one type 1 buyer.\end{lemma}
\begin{proof} Since $S$ is nonempty and every buyer has at least one incident equality edge, $\Gamma(S)$ is nonempty. The goods in $\Gamma(S)$ are completely sold and owned by the type 1 and type 3 buyers. If there are no type 3 buyers, there has to be at least one type 1 buyer. \end{proof}

For the definition of the factor $x$, we perform the following thought experiment. We increase the prices of the goods in $\Gamma(S)$ and the flow on the edges incident to them continuously by a common factor $x$ until one of four events happens: (1) a new edge enters
the equality graph or (2) the surplus of a type 2 buyer and a type 3 or 4b buyer become equal or (3) the surplus of a type 2 buyer becomes zero or (4) $x$ reaches a maximum admissible value\footnote{In~\cite{\DM}, $\xmax$ is defined as $1 + 1/(Cn^3)$. The refined definition of $S$ allows us to choose a larger value if there are no type 3 buyers. This choice will be crucial for Lemma~\ref{surplus bounds}.} $\xmax$, where
\[ \xmax = \begin{cases} 1 + \frac{1}{Cn^3}\\
\hspace{0.5em}\text{if there are type 3 buyers}\\
1 + \frac{1}{C k n^3} \\
\mbox{}\hspace{0.5em} \text{if there are no type 3's and $k$ type 1 buyers,}
\end{cases}\]
and $C = 48 e^2$.

The increase of the prices of the goods in $\Gamma(S)$ makes the
  goods in $C \setminus \Gamma(S)$ more attractive to the buyers in $S$ and hence an equality
  edge connecting a buyer in $S$ with a good in  $C \setminus \Gamma(S)$ may arise. This will happen at $x = \xeq(S)$, where 
\begin{equation*}\begin{split}
\xeq(S)=\min \{ & \frac{u_{ij}}{p_j}\cdot\frac{p_k}{u_{ik}} \; | \\
& b_i\in S,
(b_i,c_j)\in E_p, c_k\notin\Gamma(S) \} .
\end{split}\end{equation*}
%We need $O(n^2)$ multiplications/divisions to compute $\xeq(S)$.
When we increase the prices of the goods in $\Gamma(S)$ by a
common factor $x\leq \xeq(S)$, the equality edges in
$(S \times \Gamma(S)) \cup (\barS  \times (C \setminus \Gamma(S)))$ will remain in the network. In particular, by  Lemma~\ref{price increase and network}, all flow-carrying edges stay in the network.

\begin{figure}[t] 
\centerline{\includegraphics[width=0.4\textwidth]{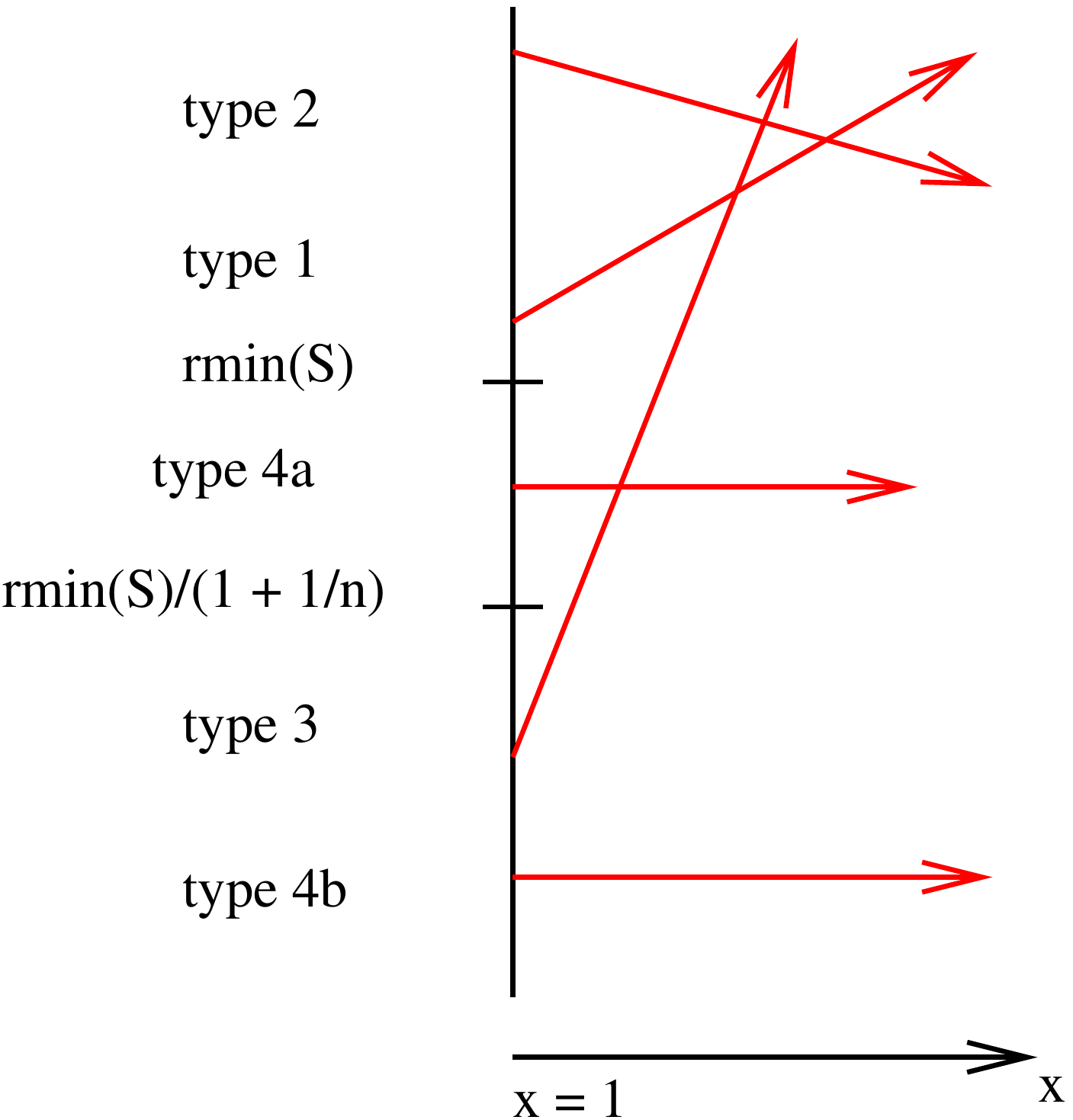}}
\caption{\label{definition of x23,x24,x2} The surpluses of various types of buyers as a function of $x$. $x_{23}$ and $x_{24}$ are the smallest values of $x$ at which the surpluses of a type 2 and a type 3 or type 4b surplus becomes equal. $x_2$ is the smallest value of $x$ at which a surplus of a type 2 buyer becomes zero. It may happen that the surplus of a type 3 buyer becomes larger than the surplus of a type 1 buyer for an $x < \min(x_{23},x_{24},x_2)$. This possibility is overlooked in~\cite{\DM} and invalidates some of their arguments. They can be fixed along the lines of this paper.}
\end{figure} 

The surplus of type 1 and 3 buyers increases,  the surplus of type 2 buyers decreases, and the surplus of type 4 buyers does not change, see Figure~\ref{definition of x23,x24,x2}. Since the total surplus does not change (recall that the surpluses of the goods are not affected by the price update), the decrease in surplus of the type 2 buyers is equal to the increase in surplus of the type 1 and 3 buyers. In particular, there are type 2 buyers. We define quantities $x_{23}(S)$ and $x_{24}(S)$ at which the surplus of a type 2 and type 3 buyer, respectively type 4b\footnote{In~\cite{\DM}, type 4 is used.} buyer, becomes equal, and a quantity $x_2$ at which the surplus of a type 2 buyer becomes zero.
\begin{align*}
x_{23}(S) &= 
 \min \{\frac{p_i + p_j - r(b_j)}{p_i + p_j - r(b_i)} \; | \\
  &\qquad\qquad\text{$b_i$ is type 2 and $b_j$ is type 3 buyer}\},
 \\
x_{24}(S) &= 
\min \{\frac{p_i - r(b_j)}{p_i - r(b_i)} \; | \\
&\qquad\qquad\text{$b_i$ is type 2 and $b_j$ is type 4b buyer}\},
\\
x_2(S)&= \min \{\frac{p_i}{p_i - r(b_i)} \; | \; \text{$b_i$ is type 2 buyer}\}. 
\end{align*}
%The quantity $x_2(S)$ is only relevant, if $S = B$. It guarantees that surpluses of buyers stay nonnegative. 

\begin{lemma}\label{new flow}
With $S$ as defined in the algorithm and $x = \min(\xmax,\xeq(S),x_{23}(S),x_{24}(S),x_2(S))$, $f'$ is a feasible flow in $N_{p'}$. \end{lemma}
\begin{proof} Obvious.
% Feasibility of $f'$ is obvious. We now show that $f'$ is a maximum flow. Step one only reroutes flow and hence the flow after step one is still maximum. If step two is performed, there is no flow from $\barS$ to $\Gamma(S)$, all nodes in $S$ have surplus, and all goods in $\Gamma(S)$ are completely sold. In step two, the flow on all edges incidents to goods in $\Gamma(S)$ is multiplied by a common factor $x$ and the flows from $s$ to goods in $S$ are changed so that flow conservation still holds. This may create new residual edges from $S$ to $s$ and from $\Gamma(S)$ to $S$. However, these edges cannot create residual paths from $s$ to $t$. Thus the flow after step two is maximum. In step three no capacities are changed and hence the flow stays maximum. 
% \marginpar{this is nonsense; I am forgetting about the new equality edges}
\end{proof}

This ends the description of the algorithm. An important property of the algorithm is that goods with nonzero surplus have price one. 

\begin{lemma}\label{pro:unchanged} Once the surplus of a good becomes zero, it stays zero. As long as a good has nonzero surplus, its price stays at one.
\end{lemma}
\begin{proof} Initially all goods have price one. The price adjustment does not change the surplus of any good and only increases the prices of goods that are completely sold. In the balanced flow $f''$, all goods that are completely sold with respect to $f$ are also completely sold with respect to $f''$. \end{proof}

\section{The Analysis}\label{sec: Analysis}

In this section, we derive a bound on the number of phases. We  distinguish between \emph{$\xmax$-phases} and \emph{balancing phases}. A phase is an $\xmax$-phase if $x = \xmax$ and is balancing otherwise. As in~\cite{\DM}, we use two potential functions for the analysis, namely the product $P = \prod_i p_i$ of all prices and the 2-norm $\norm{r(B)}$ of the surplus vector of the buyers. 

In Section~\ref{xmax-phases}, we show that the number of $\xmax$-phases is $O(n^4 \log(nU))$, and in Section~\ref{balancing phases}, we show the same bound for the number of balancing phases. This is by a factor of $n$ better than in~\cite{\DM}. 

The improvement for the number of phases comes from the more careful definition of the set $S$, the distinction between phases with and without type 3 buyers, the refined definition of $\xmax$ in phases without type 3 buyers and a refined analysis of the number of such phases (Lemma~\ref{number-of-xmax-without}), a new analysis of the number of $\xmax$-phases with type 3 buyers (Lemma~\ref{number-of-xmax-with}), a refined analysis of the norm increase in $\xmax$-phases, and an improved analysis of the norm decrease in balancing phases.

\subsection{The Number of $\xmax$-Phases.}\label{xmax-phases}

We first recall a result from~\cite{\DM} that prices stay bounded by $(nU)^n$. This immediately yields a bound of $(nU)^{n^2}$ on the product of the $x$-factors used in all phases and also on the number of $\xmax$-phases with no type 3 buyers. A different argument yields a bound on the number of $\xmax$-phases with type 3 buyers. The latter is even strongly polynomial.

\begin{lemma}[\cite{\DM}]\label{thm:largest} All prices stay bounded by $\max(n,U)^{n-1} \le (nU)^n$.
\end{lemma}
\begin{proof} The upper bound is stated as $(nU)^{n-1}$ in~\cite{\DM}. The proof actually shows the bound $\max(n,U)^{n-1}$. \end{proof}

\begin{lemma} For a phase $h$, let $x_h > 1$ be the factor by which the prices in $\Gamma(S)$ are increased.
Then 
\[    \prod_h x_h \le (nU)^{n^2}.\]
\end{lemma}
\begin{proof}
\begin{align*}
\prod_h x_h &\le \prod_{j} \left(\prod_{\text{$p_j$ is increased in phase $h$}} x_h \right)  \\
&\le \prod_j (nU)^n \le (nU)^{n^2}.
\end{align*}
\end{proof}

% \begin{lemma} The norm of surplus vector of the buyers increases by at most by a multiplicative factor $(nU)^{n^2}$ over all phases. 
% \end{lemma}
% \begin{proof} A phase increases the norm by at most by the price increase factor $x$. The product of the price increase factors is bounded by $(nU)^{n^2}$ by the preceding lemma. \end{proof}

The next two Lemmas have no equivalent in~\cite{\DM}.

\begin{lemma}\label{number-of-xmax-without} The number of $\xmax$-phases with no type 3 buyers is $O(n^4 \log (nU))$. 
\end{lemma}
\begin{proof} Let $T$ be the number of such phases. Consider any such phase and assume that there are $k$ type 1 buyers. Then the prices of exactly $k$ goods are increased by a factor $\xmax = 1 + 1/(Ckn^2)$. Thus $P = \prod_i p_i$ grows by a factor of 
\begin{align*}
(1 + \frac{1}{C k n^2})^k &= \exp(k \ln(1 + \frac{1}{C k n^2})) \\
&\ge \exp(k \frac{1}{2C k n^2}) = \exp(\frac{1}{2 C n^2}).
\end{align*}
Since $\ln P$ is bounded by $n^2 \log (nU)$, we have $T \cdot 1/(2 C n^2) \le n^2 \log (nU)$ and hence $T \le 2 C n^4 \log (nU)$. \end{proof}

\begin{lemma}\label{number-of-xmax-with} The number of $\xmax$-phases with type 3 buyers is $O(n^4 \ln n)$.\end{lemma}
\begin{proof} For $i \le 3$, let $B_i$ be the set of type $i$ buyers. We first show that the total budget of the type 3 buyers is at most the total budget of the type 2 buyers, more precisely,                                  
\[ \sum_{i \in B_3} p_i \le \sum_{i \in B_2} p_i = \sum_{i \in B_3} p_i + \sum_{i \in B_1 \cup B_2} r(b_i).\]
The goods owned by the type 1 and type 3 buyers are completely bought by the buyers of type 1 and type 2. Hence
\begin{align*}
 \sum_{i \in B_1 \cup B_3} p_i &= \sum_{i \in B_1 \cup B_3} \inflow(c_i) \\
&= \sum_{i \in B_1 \cup B_2} \outflow(b_i)\\
&= \sum_{i \in B_1 \cup B_2} (p_i - r(b_i)).
\end{align*}
Subtracting $\sum_{i \in B_1} p_i$ from both sides establishes the equality in the claim. The inequality follows since surpluses are nonnegative. 

We next show that $\sum_{i \in B_2} p_i \le 2Cn^4 $, whenever $x \ge 1 + 1/(Cn^3)$. The outflow of any type 2 buyer $b_i$ is $p_i - r(b_i)$ at the beginning of the phase. It increases by $(p_i - r(b_i))/(Cn^3)$ during the price update. The increase cannot be more than the surplus and hence $(p_i - r(b_i))/(Cn^3) \le r(b_i)$ or $p_i \le (1 + Cn^3) r(b_i)$. Summing over all type 2 buyers and observing that the total surplus is at most $n$ initially and never increases yields
\[ \sum_{i \in B_2} p_i \le (1 + Cn^3) \sum_{i \in B_2} r(b_i) \le n + Cn^4 \le 2 C n^4.\]

We finally show that the number of $\xmax$-phases with type 3 buyers is at most $n + (Cn^4 + n) \ln(2Cn^4)$. Assume otherwise. Then there must be a buyer $b_i$ such that there are at least  $1 + (Cn^3 + 1) \ln(2Cn^4)$ phases in which $b_i$ is a type 3 buyer and $x = 1 + 1/(Cn^3)$. In each such phase the price of $b_i$ increases by a factor of $1 + 1/(Cn^3)$ and hence the price of $p_i$ before the last such phase is at least 
$(1 + 1/(Cn^3))^{(Cn^3 + 1)\ln(2Cn^4)} > e^{\ln (2 Cn^4)} = 2Cn^4$. We conclude that the total budget of the buyers in $B_3$ and hence, by the first claim, the total budget of the buyers in $B_2$ exceeds $2Cn^4$. This contradicts the second claim. 
\end{proof}

\section{The Evolution of the Surplus Vector}\label{sec: Evolution}

The 2-norm of the surplus vector of the buyers is our second potential function. In the price and flow adjustment, the surpluses of type 2 and type 3 buyers move towards each other. Lemma~\ref{technical lemma} is our main tool for estimating the resulting reduction of the 2-norm of the surplus vector. Lemmas~\ref{norm increase with} and~\ref{norm increase without} show that the 2-norm of the surplus vector increases by at most a factor $1 + O(1/n^3)$ in $\xmax$-phases. Lemma~\ref{decrease in balancing} shows that a balancing phase reduces the 2-norm by a factor of $1 - \Omega(1/n^3)$. Putting everything together, we obtain the $O(n^4 \log(nU))$ bound on the number of balancing phases.

\subsection{A Technical Lemma.} 

We start with a technical lemma.

\begin{lemma}\label{technical lemma} Let $r = (r_1,\ldots,r_n)$ and $(r'_1,\ldots,r'_n)$ be nonnegative vectors.
Let $k \in [1,n]$ be such that $r'_i \ge r'_j$ for $i \le k < j$. Suppose that $\delta_i = r_i - r'_i \ge 0$ for $i \le k$ and $\delta_j = r'_j - r_j \ge 0$ for $j > k$. Let $D = \min_{i \le k} r_i - \max_{j> k} r_j$, and let $\Delta = \sum_{i \le k} \delta_i$. If $\Delta \ge 
\sum_{j > k} \delta_j$, 
\[  \norm{r'}^2 \le \norm{r}^2 - D\Delta. \]
\end{lemma}
\begin{proof}
Let $m = \max_{j > k} r_j$ and $m' = \max_{j > k} r'_j$. Then $r_i \ge m + D$ and $r'_i \ge m'$ for all $i \le k$. Therefore, 
\begin{align*}
\norm{r}^2& - \norm{r'}^2\\
&= \sum_{i \le k} \left(r_i^2 - (r_i - \delta_i)^2 \right)  + \sum_{j > k} \left( r_j^2 - (r_j + \delta_j)^2  \right)\\
&=  \sum_{i \le k} \left(2r_i   \delta_i - \delta_i^2\right) + \sum_{j > k} \left( -2r_j \delta_j - \delta_j^2 \right)\\
&=  \sum_{i \le k} \delta_i\left(r_i  +r_i - \delta_i\right) - \sum_{j > k} \delta_j \left( r_j + r_j + \delta_j\right)\\
&=  \sum_{i \le k} \delta_i (r_i + r'_i)  - \sum_{j > k} \delta_j( r_j + r'_j)\\
&\ge (m + D + m') \sum_{i \le k} \delta_i  - (m + m') \sum_{j > k} \delta_j \\
&\ge D \Delta.
\end{align*} \end{proof}

% \begin{lemma}\label{small calculation} Let $\gamma > \alpha > 0$ and assume $D = a - b \ge b/\alpha$. Then $D \ge   a/(1 + \alpha)$ and $D' = a - (1 + 1/\gamma)b \ge a/(1 + \alpha (\gamma +1) /(\gamma - \alpha))$. 
% \end{lemma}
% \begin{proof} Let $b = \gamma a$. Then $(1 - \gamma) a = D \ge b/\alpha = \gamma a/\alpha$ and hence
% $\gamma \le \alpha/(1 + \alpha)$. Thus $D \ge a (1 - \alpha/(1 + \alpha)) = a/(1 + \alpha)$. For the second claim, observe that 
% \[ D' = a - b - \frac{b}{\gamma} \ge \frac{b}{\alpha} - \frac{b}{\gamma}= \frac{\gamma - \alpha}{\alpha \gamma} \frac{1 +1/\gamma}{1 + 1/\gamma}b = \frac{\gamma - \alpha}{\alpha(\gamma+1)} (1 + 1/\gamma) b.\]
% Now apply, the first claim with $\alpha = \alpha (1 + \gamma)/(\gamma - \alpha)$. 
% \end{proof}

\subsection{$\xmax$-Phases.}

In $\xmax$-phases, the 2-norm of the surplus vector of the buyers grows. We again distinguish $\xmax$-phases with and without type 3 buyers and show that for both kind of phases the 2-norm grows by at most a factor of $(1 + O(1/n^3))$. The argument for $\xmax$-phases with type 3 buyers was essentially already given in~\cite{\DM}. The argument for $\xmax$-phases without type 3 buyers is new. For both arguments we need a bound on the maximum ratio of surpluses of buyers in $S$. In~\cite{\DM} is was shown that $\rmax(B) \le e \cdot \rmin(S)$ for their definition of $S$. For our definition of $S$, we have in addition that $\rmax(B) \le (1 + 4k/n) \rmin(S)$ in the case that there are $k$ type 1 buyers and no type 3 buyers. The latter inequality does not hold for the definition of $S$ as used in~\cite{\DM} as Figure~\ref{counterexample to old definition of S} shows; it is crucial for the improved bound on the number of phases.

\begin{figure}[t]
\centerline{\includegraphics[width=0.4\textwidth]{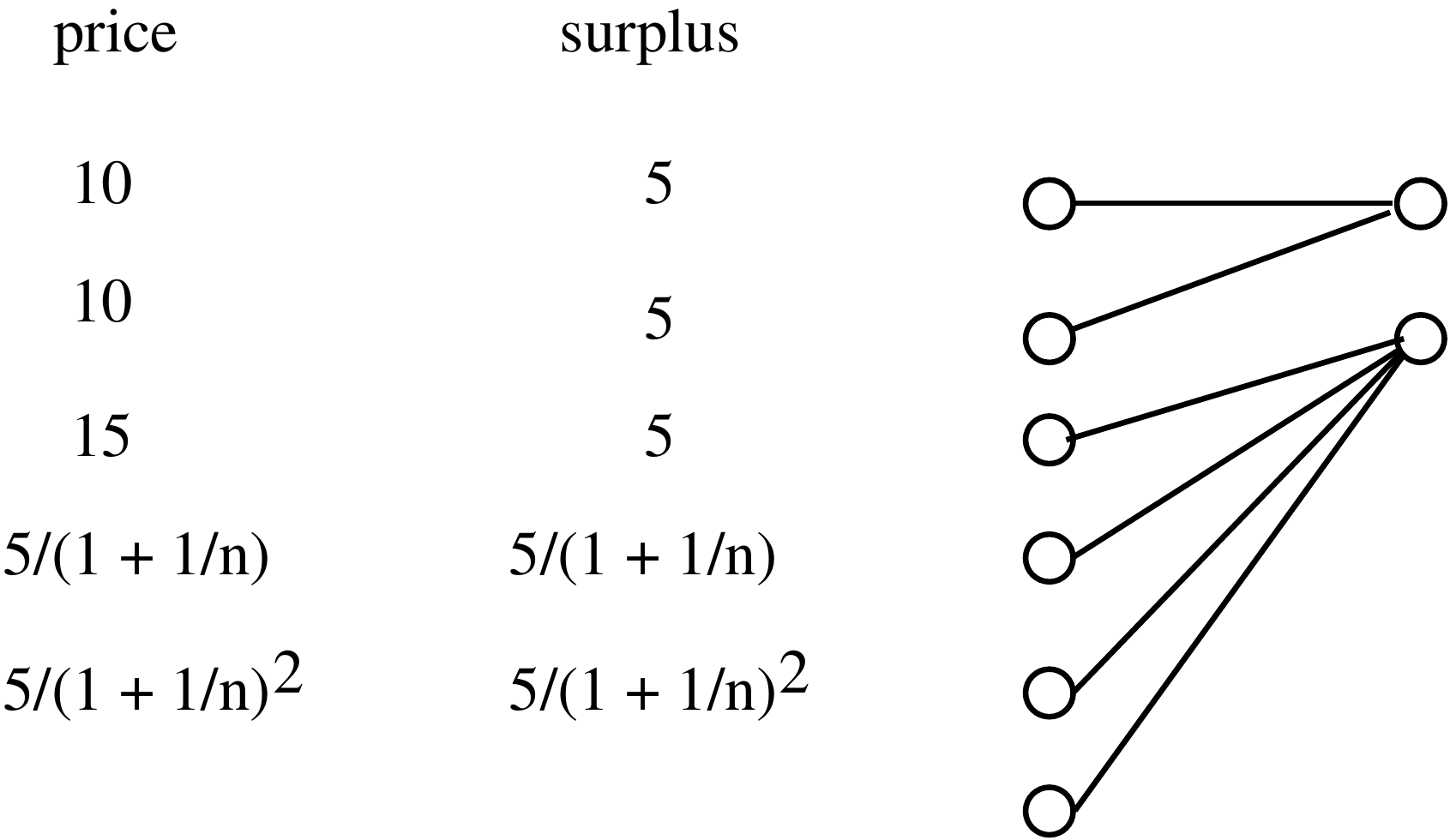}}
\caption{\label{counterexample to old definition of S}The good owned by a buyer is shown at the same height as the buyer. The edges shown are equality edges. The two topmost buyers have type 1. Their budgets are 10 and their surplus is five. Both of them spend 5 units on the good owned by the topmost buyer. The third buyer from the top has a budget of 15 and a surplus of 5. It spends 10 units on the good owned by the second buyer. Then we have a chain of arbitrary length of buyers with budget and surplus equal to $5/(1 + 1/n)^i$, $i \ge 1$. According to our definition of $S$, $S$ consists of the three topmost buyers. According to the definition used in~\cite{\DM}, the entire chain of nodes would also belong to $S$.}
\end{figure}

\begin{lemma}\label{surplus bounds} The maximum surplus of a buyer is the maximum surplus of a buyer in $S$, i.e., $\rmax(B) = \rmax(S)$. Let $k = \abs{\Gamma(S)}$. Then $\rmax(S) \le \min(e,(1 + 4k/n))\rmin(S)$. The squared norm of the surplus vector of the buyers is bounded by $n e^2 \rmin(S)^2$. If there are no type 3 buyers, $k$ is equal to the number of type 1 buyers. The maximum surplus of any buyer or good is at most $n$. \end{lemma}
\begin{proof} By definition, $S$ contains the buyer of largest surplus. 

Let $S' = \set{b_i \in S}{\outflow(b_i) > 0 \text{ or } c_i \in \Gamma(S)}$ be the subset of buyers in $S$ whose surplus changes by the price and flow adjustment. Let $r_1 > r_2 > \ldots > r_h$ be the different surplus values of the buyers in $S'$. Any two buyers in $S'$ that have positive flow to the same good have the same surplus value. Thus there are at most $k$ distinct surplus values of buyers in $S'$ with positive outflow. Every buyer $b_i$ in $S'$ either has positive outflow or $c_i \in \Gamma(S)$. There are at most $k$ buyers of the latter kind. Thus $h \le 2k$. Also, $h \le n$. 

Recall that $b_1$ is a buyer of highest surplus. We claim $b_1 \in S'$ and hence $r_1 = \rmax(S)$. Indeed, if $b_1 \not\in S'$ then $\outflow(b_1) = 0$. Let $(b_1,c)$ be an equality edge incident to $c$. Then $c$ is completely sold and the flow into $c$ comes from buyers with surplus at least the surplus of $b_1$. These buyers have outflow and hence belong to $S'$. Thus $r_1 = \rmax(S)$. 

Let $b$ be a buyer in $S'$ whose surplus is larger than $\rmin(S)$. Since $S(r(b))$ did not qualify for $S$, there must be a buyer $b'$ such that that $r(b)/(1 + 1/n) \le r(b') < r(b)$ and either $\outflow(b') > 0$ or $c' \in \Gamma(S(r(b))$. In either case, we conclude $\rmin(S) \le r(b')$ and hence $b' \in S'$. The argument also shows that $\rmin(S) = r_h$. Thus $r_{j+1} \ge r_j/(1 + 1/n)$ for $j <h$ and hence 
\begin{align*}
\frac{\rmax(S)}{\rmin(S)} &\le (1 + 1/n)^{\min(2k,n) } \le e^{\min(2k,n))/n} \\
&\le \min(e,1 + 4k/n), 
\end{align*}
where the last inequality follows from $e^x \le 1 + 2x$ for $0 \le x \le 1$. 

By the two preceding arguments, $\rmax(B) \le e \cdot \rmin(S)$. Thus the squared 2-norm of the surplus vector is at most $n e^2 \rmin(S)^2$. 

If there are no type 3 buyers, $\abs{\Gamma(S)}$ is equal to the number of type 1 buyers. 

The initial surplus of the buyers is $n$ and the total surplus never increases. Thus, no buyer can ever have a surplus of more than $n$. 
\end{proof}

We turn to $\xmax$-phases with type 3 buyers.

\begin{lemma}\label{norm increase with}  In a $\xmax$-phase with type 3 buyers, the 2-norm of the surplus vector of the buyers increases by at most a factor $1 + O(1/n^3)$. The total multiplicative increase of the 2-norm of the surplus vector of the buyers over all $\xmax$-phases with type 3 buyers is $n^{O(n)}$. 
\end{lemma}
\begin{proof} Consider the price and flow adjustments in an $\xmax$-phase with type 3 buyers, and let $r'$ be the surplus vector with respect to the flow $f'$. The surpluses of the type 1 buyers are multiplied by the factor $1 + 1/(Cn^3)$, the surpluses of the type 2 buyers go down, the surpluses of the type 3 buyers go up, and the surpluses of the type 4 buyers do not change. Also, the surpluses of the type 2 buyers do not fall below the surpluses of the type 3 buyers. The total decrease of the surpluses of the type 2 buyers is equal to the total increase of the surpluses of type 1 and type 3 buyers and hence is at least the total increase of the surpluses of the type 3 buyers. 

We split the surplus vector into three parts: the surpluses $r_1$ corresponding to type 1 buyers, the surpluses $r_{2,3}$ corresponding to the type 2 and 3 buyers, and the surpluses $r_4$ corresponding to type 4 buyers. The norm of the first subvector is multiplied by $x$, the norm of the third subvector does not change, and we can bound the change in the norm of the second subvector by Lemma~\ref{technical lemma}. In particular, the norm of the second subvector does not increase. Thus
\begin{align*} 
\norm{r'}^2 - \norm{r}^2 &=  (\norm{r'_1}^2 - \norm{r_{1}}^2) + (\norm{r'_{23}}^2 - \norm{r_{23}}^2)\\
&\mbox{}\quad + (\norm{r'_4}^2 - \norm{r_4}^2)\\
&\le ((1 + 1/(Cn^3))^2 - 1) \cdot \norm{r_1}^2\\
&\le 3/(Cn^3) \cdot \norm{r}^2,
\end{align*}
and hence 
\[ \norm{r'} \le \sqrt{1 + 3/(Cn^3)} \norm{r} \le (1 + 3/(Cn^3)) \norm{r}.\]The 2-norm of the surplus vector with respect to $f''$ is at most $\norm{r'}$. 

The number of $\xmax$-phases with type 3 buyers is $O(n^4 \ln n)$. Hence the total multiplicative increase in $\xmax$-phases with type 3 buyers is at most
\[  (1 + \frac{3}{Cn^3})^{O(n^4 \ln n)}  \le e^{O(n \ln n)} = n^{O(n)}.\]
\end{proof}

We next turn to $\xmax$-phases without type 3 buyers.

\begin{lemma}\label{norm increase without} Consider an $\xmax$-phase with no type 3 buyer. Then $\norm{r'} \le (1 + O(1/n^3)) \norm{r}$. The total multiplicative increase of the norm of surplus vector in $\xmax$-phases with no type 3 buyers is $(nU)^{O(n)}$. 
\end{lemma}
\begin{proof} Let $k$ be the number of type 1 buyers and let $R =\rmin(S)$ be the minimum surplus of any buyer in $S$. By Lemma~\ref{surplus bounds}, the surplus of any buyer is at most $(1 + 4k/n)R$. Also, $\xmax = 1 + 1/(Ck n^2)$. 

The increase of the surpluses of the type 1 buyers is equal to the decrease of the surpluses of the type 2 buyers. Let $\delta_i$ be the change of surplus of buyer $i$, $i \in S$, and let $\Delta$ be the total increase in surplus of the type 1 buyers. Let $B_1$ and $B_2$ be the type 1 and 2 buyers. Then
\begin{align*}
\norm{r'}^2 - \norm{r}^2 &= \sum_{i \in B_1} (r_i + \delta_i)^2 - r_i^2 + \sum_{i \in B_2}(r_i - \delta_i)^2 - r_i^2 \\
&\le 2 \Delta ( (1 + 4k/n) R  - R) + 2 \Delta^2.
\end{align*}
Next observe that $\Delta \le (\xmax - 1) k (1 + 4k/n) R \le \frac{1}{C k n^2} \cdot k \cdot 5 \cdot R \le 5R/(C n^2)$. Thus
\[ \norm{r'}^2 - \norm{r}^2 \le O(\frac{R}{n^2} \frac{kR}{n} + \frac{R^2}{n^4}) = O(\frac{kR^2}{n^3}) =  O(\frac{1}{n^3}) \norm{r}^2, \]
since $\norm{r}^2 \ge k R^2$. Thus $\norm{r'} \le \sqrt{1 +O(1/n^3)} \norm{r} = (1 + O(1/n^3)) \norm{r}$. 

The number of $\xmax$-phases with no type 3 buyers is $O(n^4 \log (nU))$. Hence the multiplicative increase is bounded by 
\[   ( 1 + O(\frac{1}{n^3}))^{n^4 \log (nU)} = \exp(O(n \log (nU))) = (nU)^{O(n)}.\]
\end{proof}

\subsection{Balancing Phases.}\label{balancing phases}

In a balancing phase, $x < \xmax$. Our goal is to show that balancing phases reduce the norm of the surplus vector of the buyers by a factor $1 - \Omega(1/n^3)$. 

\begin{lemma}\label{decrease in balancing} Let $r''$ be the surplus vector with respect to $f''$. In balancing phases,
\[   \norm{r''} \le (1 - \Omega(\frac{1}{n^3})) \norm{r}.\]
\end{lemma}
\begin{proof} In a balancing phase, we first increase the prices of the goods in $\Gamma(S)$ and the flows into these goods by a common factor $x = \min(\xeq, x_{23}, x_{24}, x_2)$ to obtain a flow $f'$ and then construct a balanced flow $f''$ in the network $N_{p'}$. The 2-norm of $f''$ is certainly no larger than the 2-norm of $f'$. 

We have $\rmin(S)/(1 + 1/n) \ge \rmax(B_3 \cup B_{4b})$ (equivalently $\rmin(S) - \rmax(B_3 \cup B_{4b}) \ge \rmin(S)/(n+1)$) at the beginning of the phase. The price and flow update affects this gap. We distinguish cases according to whether $f'$ closes half of the gap or not. We introduce $R = \rmin(S)$ as a shorthand. \medskip

\paragraph{Case One, $\rmin'(S) - \rmax'(B_3 \cup B_{4b}) < \rmin(S)/(2(n+1))$:} This case occurs when $x = x_{23}$ or $x = x_{24}$ or $x = x_2$. It can also occur, when $x = \xeq$. Recall, that the total decrease of the surpluses of the type 2 buyers is equal to the total increase of the surpluses of the type 1 and type 3 buyers, that surpluses of type 1 and type 3 buyers do not decrease and that surpluses of type 2 buyers do not increase. Thus, it must be the case that the total decrease of the surpluses of the type 2 buyers plus the total increase of the surpluses of the type 3 buyers is at least $\rmin(S)/(2(n+1))$. 

We split the surplus vector into three parts: the surplus vector $r_1$ corresponding to type 1 buyers, the surplus vector $r_{2,3}$ corresponding to the type 2 and 3 buyers, and the surplus vector $r_4$ corresponding to type 4 buyers. The 2-norm of the third subvector does not change. The 2-norm of the first subvector increases. If there are type 3 buyers, $x \le 1 + 1/(Cn^3)$ and hence 
\[ \norm{r'_1}^2 - \norm{r_{1}}^2 \le  ((1 + \frac{1}{Cn^3})^2 - 1) \cdot n (eR)^2 \le \frac{3 e^2 R^2}{Cn^2}.\]
If there are no type 3 buyers and $k$ type 1 buyers, $x \le 1 + 1/(C k n^2)$ and hence 
\[ \norm{r'_1}^2 - \norm{r_{1}}^2 \le  ((1 + \frac{1}{C k n^2})^2 - 1) \cdot k (eR)^2 \le \frac{3 e^2 R^2}{Cn^2}.\]
We next bound the change in the norm of the second subvector by Lemma~\ref{technical lemma}. Let $D = \rmin(B_2) - \rmax(B_3)$ be the minimum distance between a type 2 and a type 3 surplus, and let $\Delta$ be the total decrease of the surpluses of the type 2 buyers. Then $D \ge R/(n+1)$ and $\Delta \ge R/(4(n+1))$. Thus the squared 2-norm of the the vector $r_{23}$ decreases by at least $R^2/(4(n+1)^2)$. Thus
\begin{align*} 
\norm{r'}^2 - \norm{r}^2 &=  (\norm{r'_1}^2 - \norm{r_{1}}^2) + (\norm{r'_{23}}^2 - \norm{r_{23}}^2) \\
&+ (\norm{r'_4}^2 - \norm{r_4}^2)\\
&\le \frac{3 e^2 R^2}{Cn^2} - \frac{R^2}{4(n+1)^2}\\
&\le - \Omega(\frac{1}{n^3}) \norm{r}^2\qquad\text{since $\norm{r}^2 \le n e^2 R^2$}
\end{align*}
and hence $\norm{r'} \le (1 - \Omega(\frac{1}{n^3}))\norm{r}$. \medskip

\paragraph{Case Two, $\rmin'(S) - \rmax'(B_3 \cup B_{4b}) \ge \rmin(S)/(2(n+1))$:} This can only be the case if $x = \xeq$. We first observe that 
\begin{align*}
 \norm{r''}^2 - \norm{r}^2 &=  \norm{r''}^2 - \norm{r'}^2 + \norm{r'}^2 - \norm{r}^2 \\
&\le 
 \norm{r''}^2 - \norm{r'}^2 + \frac{3 e^2 R^2}{C n^2},\end{align*}
where the inequality follows from the proof of first case. We need to relate the squared 2-norm of the surplus vectors with respect to the flows $f''$ and $f'$. To this end, we construct a feasible flow $\hf$ from $f'$ and argue about the 2-norm of its surplus vector. Since $f''$ is a balanced flow, its 2-norm is at most the 2-norm of $\hf$. The construction of $\hf$ is described in Figure~\ref{construction of hf}.

\begin{figure*}[t]
\centerline{\framebox{\parbox{\textwidth}{
\begin{tabbing}
555\=555\=555\=555\=\kill
\>Let $\Enew$ be the set of new equality edges from $S$ to $C \setminus \Gamma(S)$;\\[0.3em]
\> Let $\hf = f'$;\\[0.3em]
\> If $\rmin(S,\hf) \le \rmax(B_3 \cup B_{4b},\hf)$, terminate the construction of $\hf$;\\[0.3em]
\> For every edge $(b_i,c_j) \in \Enew$ do\\[0.3em]
\>\>\parbox{0.85\textwidth}{As long as $c_j$ is not completely sold, augment along $(b_i,c_j)$ gradually
until $\rmin(S,\hf) = \rmax(B_3 \cup B_{4b},\hf)$ or $c_j$ is completely sold;} \\[0.3em]
\>\> In the former case, terminate the construction;\\[0.3em]
\>\>For all flow-carrying edges $(b_k,c_j)$ with $b_k \in \barS$:\\[0.3em]
\>\>\>\parbox{0.8\textwidth}{Augment along $(b_i,c_j,b_k)$ gradually until 
$r_{\min}(S,\hf) = r_{\max}(B_3 \cup B_{4b},\hf)$ or the flow from $b_k$ to $c_j$ becomes zero;} \\[0.3em]
\>\>\>In the former case, terminate the construction;
\end{tabbing} \vspace{-1em}
}}} \caption{New equality edges: the construction of $\hf$ from $f'$.}\label{construction of hf}
\end{figure*}

We initialize $\hf$ to $f'$, the flow after price and flow adjustment. Let (*) be a shorthand for $\rmin(S,\hf) \le \rmax(B_3 \cup B_{4b},\hf)$'. Let $\Enew$ be the new equality edges connecting buyers in $S$ with goods in $C \setminus \Gamma(S)$. There may also be new equality edges\footnote{The buyers incident to such equality edges had only equality edges connecting them to $\Gamma(S)$ before the price adjustment. These equality edges carried no flow.} from $\barS$ to $C \setminus \Gamma(S)$; they are of no concern to us. We iterate over the edges $(b_i,c_j)$ in $\Enew$. For any such edge, we first increase the flow across the edge until either $c_j$ is completely sold or (*) holds. In the latter case, the construction of $\hf$ is complete. In the former case, we
consider the flow-carrying edges $(b_k,c_j)$ with $b_k \in \barS$. Since type 4a nodes have no outgoing flow, $b_k$ has type 3 or type 4b and hence $r(b_k) \le r(b_i)/(1 + 1/n)$. We increase the flow along the edge $(b_i,c_j)$, decrease the flow across the edge $(b_k,c_j)$ by the same amount until either the latter flow is zero or (*) holds. In the latter case, the construction of $\hf$ is complete. 

\begin{lemma} When the algorithm in Figure~\ref{construction of hf} terminates, either condition (*) holds for $\hf$ or there is at least one buyer in $S$ whose surplus with respect to $\hf$ is one less than its surplus with respect to $f'$.\end{lemma}
\begin{proof} Consider any new equality edge $(b_i,c_j)$. If condition (*) does not hold for $\hf$, then all inflow to $c_j$ in $\hf$ comes from $b_i$ and $c_j$ is completely sold. In $f$ there was no flow from $b_i$ to $c_j$ and hence the surplus of $b_i$ with respect to $\hf$ is one less than its surplus with respect to $f'$. \end{proof}

In the construction of $\hf$ from $f'$, the surpluses of type 1 and 2 players do not increase, and the surpluses of type 3 and 4 players do not decrease. Also, the total decrease of the former surpluses is at least the total increase of the latter surpluses. Thus Lemma~\ref{technical lemma} applies. We have $D \ge R/(2(n+1))$ and $\Delta \ge D/2$ if condition (*) holds for $\hf$ and $\Delta \ge 1$ otherwise. Thus
\begin{align*}
\norm{\hr}^2 - \norm{r'}^2 & \le - \frac{R}{2(n+1)} \min(\frac{R}{4(n+1)},1) \\
&= -\frac{R^2}{8(n+1)^2},\end{align*}
since $R \le n$ by Lemma~\ref{surplus bounds}.

Combining the bounds yields
\begin{align*}
\norm{r''}^2 - \norm{r}^2 &\le 
 \norm{\hr}^2 - \norm{r'}^2 + \frac{3 e^2 R^2}{C n^2} \\
&\le -\frac{R^2}{8(n+1)^2} + 
\frac{3 e^2 R^2}{C n^2} \\
&= - \Omega(\frac{1}{n^3}) \norm{r}^2.
\end{align*}
\end{proof}

We can now prove a bound on the number of balancing  phases.

\newcommand{\onenorm}[1]{\norm{#1}_1}

\begin{lemma}\label{lem:number-of-good} The number of balancing phases is $O(n^4 \log (nU))$.\end{lemma}
\begin{proof} Let $T$ be the number of balancing phases. The 2-norm of the initial balanced flow is no larger than $\sqrt{n}$. In the $\xmax$-phases, the 2-norm of the surplus vector is multiplied by a factor of at most $(nU)^{O(n)}$. We exit from the loop once $\abs{r}<\epsilon$. Thus $\abs{r} \ge \epsilon$ and hence  $\norm{r} \ge \epsilon/\sqrt{n}$ after $T -1$ 
balancing phases. Each balancing phase reduced the 2-norm by a factor of $1 - \Omega(1/n^3)$. Therefore
\[   \sqrt{n} (nU)^{O(n)} ( 1 - \Omega(1/n^3))^{T - 1} \ge \epsilon/\sqrt{n}.\]
The bound follows. \end{proof}

% \begin{lemma} Assume that $E_p$ is a tree, all prices are positive, and $(b,c)$ is an edge of the equality graph with at least one endpoint having degree one. Then no minimum cut contains $(s,b)$ and $(c,t)$. \end{lemma}
% \begin{proof} Assume $b$ has degree one. If $C$ is a cut containing $(s,b)$ and $(c,t)$ then $C \setminus (s,b)$ is also a cut. Its capacity is smaller. \end{proof}

\section{Balanced Flows for Nondegenerate Instances}\label{sec: Balanced Flows}

We show that for nondegenerate instances, balanced flows can be computed with one maxflow computation and $n$ maxflow computations in forest networks. It is known~\cite{DPSV08} that balanced flows in equality networks can be computed with $n$ maxflow computations. All computations are in graphs with $O(n)$ edges. 

\newcommand{\capa}{\mathit{cap}} \newcommand{\val}{\mathit{val}}

\begin{lemma} If $E_p$ is a forest, a maximum flow in $N_p$ can be computed with $O(n)$ arithmetic operations. 
\end{lemma}
\begin{proof} For an edge $e$, we use $\capa(e)$ to denote its capacity. The following algorithm constructs a maximum flow: As long as $E_p$ is nonempty, let $(b_i,c_k)$ be an equality edge such that either $b_i$ or $c_k$ has degree one in $E_p$. We route $q = \min(\capa((s,b_i),\capa(c_k,t))$ along the path $(s,b_i,c_k,t)$, reduce the capacities of the edges $(s,b_i)$ and $(c_j,t)$ by $q$, and remove the edge $(b_i,c_k)$ from the network. 
Let $(N',\capa')$ be the resulting network.

In order to show correctness, we construct a cut $D$ such that the value of the flow constructed by the algorithm is equal to the capacity of $D$. We do so by induction on the number of edges in $E_p$. If $E_p$ is empty, we take $D = s \cup B$. Assume inductively, that $f'$ is a flow in $(N',\capa')$ and $D'$ is a cut in this network with $\val(f') = \capa'(D')$. Let $f$ be obtained from $f'$ by routing $q$ units along the path $(s,b_i,c_k,t)$. Then $\val(f) = \val(f') + q = \capa'(D') + q$. If $b_i$ and $c_k$ belong to the same side of the cut (either both belong or neither belongs to $D'$), then $\capa(D') = \capa'(D') + q$, and we take $D = D'$. If $b_i$ and $c_k$ belong to different sides of the cut,
we obtain $D$ from $D'$ by moving the vertex of degree one to the other side. Then exactly one of $(s,b_i)$ and $(c_k,t)$ is cut by $D$, $(b_i,c_k)$ is not cut, and no other equality edge is cut, since we move the vertex of degree one to the other side. Thus $\capa(D) = \capa'(D') +q$.
\end{proof} 

% Consider any maximum flow $f$. Since $f$ is a maximum flow either the edge $(s,b_i)$ or the edge $(s,c_j)$ is saturated; assume the former. If $b_i$ has degree one, the entire flow from $s$ to $b_i$ is routed across $(s,c_j)$ and hence $q = \capa(s,b_i)$. So assume that $b_i$ has degree larger than one, Then $c_k$ has degree one and hence the flow across $(c_k,t)$ is equal to the flow on $(b_i,c_k)$. We increase the flow on $(b_i,c_k,t)$ to $q$ and reduce the flow on other paths from $b_i$ to $t$ by the same amount. Thus there is a maximum flow which routes $q$ units along the path $(s,b_i,c_k,t)$. An induction argument finishes the proof. \end{proof}

\begin{lemma} If $E_p$ is a forest, a balanced flow in $N_p$ can be computed with $O(n^2)$ arithmetic operations. 
\end{lemma}
\begin{proof} A balanced flow can be computed with at most $n$ maxflow computations. The maxflow computations are on a network with reduced capacities and some edges removed. Thus, the set of edges connecting buyers and goods always form a forest. 
\end{proof}

\begin{lemma}\label{balanced flow acyclic}  Let $C_0 \subseteq C$ be a set of goods. If $E_p$ is a forest and there is a maximum flow in which all goods in $C_0$ are completely sold, then a balanced flow in which all goods in $C_0$ are completely sold can be computed with $O(n^2)$ arithmetic operations. 
\end{lemma}
\begin{proof} Let $f$ be a balanced flow in $N_p$. Let $B_1$, $B_2$, \ldots, $B_h$ be the partition of the buyers into maximal classes with equal surplus. This partition is unique. Let $r_i$ be the surplus of the buyers in $B_i$ with $r_1 > r_2 > \ldots > r_h \ge 0$. If $r_h = 0$, let $h_0 = h-1$. Otherwise, let $h_0 = h$. Let $C'$ be the set of goods to which money from a buyer in $B' = \cup_{i \le h_0} B_i$ flows. The goods in $C'$ are completely sold, there is no money flowing from $B_{h_0 + 1}$ ($B_{h +1} = \emptyset$) to $C'$ and there is no equality edge from $B'$ to $C \setminus C'$. Thus $s \cup B' \cup C'$ is a minimum cut. 

Let $f'$ be a maximum flow in which all goods in $C_0$ are completely sold. Then $f'$ also sends no flow from $B_{h_0 + 1}$ to $C'$ and $f'$ saturates all edges in the minimum cut. The flow $f''$ consisting of $f$ on the edges in the subnetwork spanned by $s \cup B' \cup C' \cup t$ and of $f'$ on the edges in $s \cup B_{h_0 + 1} \cup C \setminus C' \cup t$. It is a balanced flow in which all goods in $C_0$ are completely sold. 

We can construct $f''$ as follows. We first construct $f$. This requires $O(n^2)$ arithmetic operations by the preceding Lemma and gives us $B_{h_0+1}$ and $C'' \define C \setminus C'$.  Let $F$ be the value of a maximum flow in $s \cup B_{h_0 + 1} \cup C'' \cup t$ and let $P$ be the total price of the goods in 
$C'' \cap C_0$. We introduce a new vertex $t'$, replace $t$ by $t'$ in all edges from $C''  \setminus C_0$ to $t$, and introduce an edge $(t',t)$ with capacity $F - P$, and compute a maximum flow in the network $s \cup B_h \cup C'' \cup t' \cup t$. In this flow the goods in $C'' \cap C_0$ are completely sold. \end{proof}

\begin{theorem}\label{theorem arithmetic} For nondegenerate instances, the algorithm computes equilibrium prices with $O(n^6 \log (nU))$ arithmetic operations on rationals.
\end{theorem} 
\begin{proof} The number of phases is $O(n^4 \log (nU))$ and each phase requires $O(n^2)$ arithmetic operations.
Part B requires $O(n^4 \log (nU))$ arithmetic operations. \end{proof}

\section{Perturbation}\label{sec: Perturbation}

A perturbation of the utilities replaces each nonzero utility $u_{ij}$ by a nearby utility $\tu_{ij}(\epsilon)$, where $\epsilon$ is a positive parameter and $\lim_{\epsilon \rightarrow 0+} \tu_{ij}(\epsilon) = u_{ij}$. We will usually write $\tu_{ij}$ instead of $\tu_{ij}(\epsilon)$. The solution to a perturbed problem may or may not give information about the solution to the original problem. An example, where it does not is the following version of the convex hull problem. We are given a set of points in the plane; the task is to compute the number of points on the boundary of the convex hull (if points lie on the interior of hull edges, this number may be more than the number of vertices of the convex hull). A perturbation may move a point from the boundary into the interior and hence the number computed for the perturbed problem may be smaller than the number for the original problem. For the computation of equilibrium prices the situation is benign and the limit of solutions to the perturbed problem is a solution to the original problem. 

\newcommand{\tA}{\tilde{A}} \newcommand{\tX}{\tilde{X}}

An equilibrium price vector $p$ gives rise to a $n \times n$ linear system $A x = X$ of equations of full rank; the price vector $p$ is the unique solution to this system. The entries of $A$ depend on the utilities $u_{ij}$ and $X$ is a unit vector. The solution $\tp$ to a perturbed problem gives rise to a linear system $\tA x = X$. The entries of $\tA$ depend on the perturbed utilities $\tu_{ij}$. We will show in Section~\ref{Continuity for Sufficiently Small Perturbations} that for sufficiently small perturbations, a system $Ax = X$ with solution $p$ can be constructed from the system $\tA x = X$ by simply reversing the perturbation, i.e., by replacing occurrences of $\tu_{ij}$ by $u_{ij}$. This implies 
$\lim_{\epsilon \rightarrow 0+} \tp(\epsilon) = p$. Once $p$ is known, the computation of the money flow is tantamount to solving a maxflow problem. 

In the Appendix, we will review the perturbation suggested in~\cite{Orlin10} and argue that the description is incomplete. In Section~\ref{ssec: Our Perturbation}, we introduce a novel perturbation. It is inspired by~\cite{Seidel:Perturbation}. In Section~\ref{ssec: Implementation of a Phase} we review how phases are implemented in~\cite{\DM}. Finally, in Section~\ref{ssec: Degeneracy Removal}, we combine the ingredients and show how to handle degenerate inputs efficiently.

\subsection{Continuity for Sufficiently Small Perturbations.}\label{Continuity for Sufficiently Small Perturbations}

\newcommand{\hE}{\hat{E}} %\newcommand{\hG}{\hat{G}}

Let $p$ be a set of prices. The extended equality graph $\hG_p$ has edge set $E_p \cup \set{(b_i,c_i)}{1 \le i \le n}$. 
We call a set of prices \emph{canonical} if the extended equality graph is connected and the minimum price is equal to one. 

\begin{lemma} Let $p$ be any equilibrium price vector. Then there is a canonical equilibrium price vector $\hp$ such that $E_p \subseteq E_{\hp}$. \end{lemma}
\begin{proof} Let $f$ be a feasible money flow with respect to $p$. We first multiply all prices and flows by a suitable constant so as to make all prices at least one. As long as the extended equality graph contains more than one component, we choose one of them subject to the constraint that it is not the only component that contains a good with price one. Let $K$ be the component chosen. We increase the prices of all goods in $K$ and the flows between the buyers and goods in $K$ by a common factor $x$ until a new equality edge between a buyer in $K$ and a good outside $K$ arises. This reduces the number of components. Continuing in this way, we obtain a canonical equilibrium  price vector and the corresponding money flow. 
\end{proof}

Let $p$ be a canonical equilibrium price vector and let $T$ be a spanning forest of the equality graph $G_p$. Consider the following system of equations. For a buyer $b_i$ with neighbors $c_{j_1}$ to $c_{j_k}$ in $T$ we have equations 
\begin{equation}\label{forest equations} u_{i j_1} p_{j_\ell} - u_{i j_\ell} p_{j_1} = 0   \quad\text{for $2 \le \ell \le k$.} \end{equation}
For all but one  component $K$ of the equality graph (not the extended equality graph), we have the equation
\begin{equation}\label{component equations} \sum_{b_i \in B \cap K} p_i = \sum_{c_j \in C \cap K} p_j .\end{equation}
Note that $\set{i}{b_i \in B \cap K} \not= \set{j}{c_j \in C \cap K}$ since otherwise $K$ would be a component of the extended equality graph. If the equality graph has only one component, there is no equation of type (\ref{component equations}). Finally, we have the equation 
\begin{equation}\label{price one equation} p_i = 1 \end{equation} for one of the goods of price $1$. 

\begin{lemma} The equation system (\ref{forest equations}), (\ref{component equations}), (\ref{price one equation}) has full rank and $p$ is the unique solution.\end{lemma}
\begin{proof} Shown in~\cite{\DM}.\end{proof}

Let $\tu$ be a set of perturbed utilities and let $\tp$ be a canonical solution with respect to $\tu$. Then $\tp$ satisfies
a system of equations consisting of equations of the form  (\ref{forest equations}), (\ref{component equations}) and  (\ref{price one equation}). The utilities in the equations of type  (\ref{forest equations}) are the perturbed utilities $\tu$. This is system $\tA x = X$; it has solution $\tp$. We next construct a system $A x = X'$, still with solution $\tp$, by replacing the equations of type (\ref{forest equations}) by the equivalent equations
\begin{equation}\label{modified forest equations} 
 u_{i j_1} p_{j_\ell} - u_{i j_\ell} p_{j_1} = (u_{i j_1} - \tu_{i j_1}) p_{j_\ell} + (\tu_{i j_\ell} - u_{i j_\ell}) p_{j_1}
\end{equation}
and keeping the other equations. 

\newcommand{\pstar}{p^*}

\begin{lemma} The system (\ref{modified forest equations}), (\ref{component equations}), (\ref{price one equation}) has full rank. Let $p^*$ be a solution to the system $A x = X$ obtained from $Ax = X'$ by replacing the right hand side of equations (\ref{modified forest equations}) by zero. If $\delta = \max_{ij} \abs{\tu_{ij} - u_{ij}}$ satisfies $2 U 2 (nU)^{2n} + 2 ( (nU)^n + 6 (nU)^{2n} \delta) \delta < 1$ and $2 n (nU)^{2n}\delta < 1/2$, then $p^*$ is an equilibrium price vector for the original utilities.
\end{lemma}
\begin{proof} We need to show that $(s, B \cup C \cup t)$ is a minimum cut in the equality network $N_{p^*}$. To this end, we first show that equality edges with respect to $\tp$ are also equality edges with respect to $p^*$, i.e., $E_{\tp} \subseteq E_{p^*}$. 

The following argument was essentially given in~\cite{\DM}, although in a different context. Since prices are bounded by $(nU)^n$, the right hand sides in equations (\ref{modified forest equations}) are bounded by
$2 (nU)^n \delta$. The vector $p^*$ is a rational vector with a common denominator $D = \abs{\det A}$, say $p^*_i = q_i/D$ with $q_i \in \N_{> 0}$. Since $\norm{X - X'}_\infty \le 2 (nU)^n \delta$, we have $\tp - p^*= A^{-1} (X - X') = B^{-1}(X - X')/(\det A)$,
where $B^{-1}$ is an integral matrix whose entries are bounded by $n!U^n$. Thus
\[  \abs{\tp_i - p^*_i} = \abs{\tp_i - q_i/D } \le \frac{(nU)^n 2 (nU)^n \delta}{D} = \frac{2 (nU)^{2n} \delta}{D}.\]
Let $\varepsilon' = 2 (nU)^{2n}\delta$. Consider any $b_i \in B$ and $c_j,c_k \in C$ and assume $\tilde{u}_{ij}/\tp_j \le \tilde{u}_{ik}/\tp_k$. Then
\begin{align*} u_{ij} q_k &\le (\tu_{ij} + \delta) (\tp_k D + \varepsilon')\\
&\le \tu_{ij} \tp_k D + \delta(\tp_k D + \varepsilon') + \tu_{ij} \varepsilon'\\
&\le \tu_{ik} \tp_j D +  \delta(\tp_k D + \varepsilon') + \tu_{ij} \varepsilon'\\
&\le (u_{ik} + \delta)(q_j + \varepsilon') +  \delta(q_k  + 2\varepsilon') + \tu_{ij} \varepsilon'\\
&\le u_{ik} q_j + u_{ik} \varepsilon' + \delta (q_j + \varepsilon') +  \delta(q_k + 2\varepsilon') + \tu_{ij} \varepsilon'\\
&< u_{ik} q_j + 1,
\end{align*}
since $u_{ik} \varepsilon' + \delta (q_j + \varepsilon') +  \delta(q_k + 2 \varepsilon') + \tu_{ij} \varepsilon' \le 
(2 U 2 (nU)^{2n} + 2 ( (nU)^n + 6 (nU)^{2n} \delta)) \delta < 1$. 

Let $Q = \sum_i q_i$. The cut $(s, B \cup C \cup t)$ is a minimum cut in $N_{\tp}$. Its capacity is at least 
$\sum_i \tp_i \ge (Q - n \varepsilon')/D$. The cut $(s,B \cup C \cup t)$ has capacity $Q/D$ in $N_{p^*}$. If it is not a minimum cut, there is a cut of capacity at most $(Q - 1)/D$. Since $E_{\tp} \subseteq E_{p^*}$, this cut is also a cut in $N_{\tp}$. Its capacity in $N_{\tp}$ is at most $(Q - 1 + n \varepsilon')/D$. This is less than $(Q - n \varepsilon')/D$, a contradiction. 
\end{proof} 

%{\bf Need $\delta <  1/ 5 U (nU)^{2n}$. }

\subsection{Our Perturbation.}\label{ssec: Our Perturbation}

R.~Seidel~\cite{Seidel:Perturbation} discusses the nature and meaning of perturbations in computational geometry. His advice is also applicable here. First, determine a set of inputs which are nondegenerate, i.e., a set of utilities for which 
\[   U(D) \define \frac{\prod_{e \in D_1} u_e}{\prod_{e \in D_0} u_e} \not= 1 \]
for every cycle $D$; here $D_0$ and $D_1$ are the two classes of $D$. This is easy. We choose for each edge $e = ij$ a distinct prime $q_e$. Then 
\[   Q(D) \define \frac{\prod_{e \in D_1} q_e}{\prod_{e \in D_0} q_e} \not= 1 \]
for every cycle $D$. Second, one defines the perturbed utilities by combining the original utilities and the nondegenerate utilities by means of a positive infinitesimal $\varepsilon$. Since our quantity of interest combines the utilities by multiplications and divisions, a multiplicative combination is appropriate, i.e., we define
\[   \tu_{ij}(\varepsilon) = u_{ij} q_{ij}^{\varepsilon}.\]

\begin{lemma} 
For a cycle $D$, let 
\[ \tU(D) \define \frac{\prod_{e \in D_1} \tu_e}{\prod_{e \in D_0} \tu_e} = U(D) Q(D)^\varepsilon.\]
Then, for all sufficiently small positive $\varepsilon$, $\sign(\ln \tU(D)) \not= 0$ and $\sign(\ln \tU(D)) = \sign(\ln U(D))$ if $\sign(\ln U(D)) \not = 0$. 
\end{lemma}
\begin{proof} 
Clearly, $\ln \tU(D) = \ln U(D) + \varepsilon \ln Q(D)$. If $U(D) \not= 1$, $\sign( \ln \tU(D)) = \sign(\ln U(D))$ for every sufficiently small $\varepsilon$. If $U(D) = 1$, 
$\sign( \ln \tU(D)) = \sign(\ln Q(D))$ for every positive $\varepsilon$. 
\end{proof}

Since there are at least $Q/(2 \ln Q)$ primes less than $Q$ for any integer $Q$, we can choose the $q_e$'s such that $q_e \le Q \define 8n^2\log n$. The primes less than $Q$ can be determined in time $O(Q \ln Q) = O(n^2 \log^2 n)$ by use of the Sieve of Eratosthenes.\footnote{We initialize all entries of an array $A[2..Q]$ to true. Then, for $i \in [2,Q]$, we do: if $A[i]$ is true, we output $i$ as prime and set all multiples $A[\ell\cdot i]$, for $2 \le \ell \le Q/i$ to false. All of this takes time
$O(\sum_{i \le Q} Q/i) = O(Q \ln Q)$.}

We come to the implementation of the perturbation. The perturbed utilities are of the form $A B^\varepsilon$ with $A$ and $B$ positive rational numbers. We will also restrict prices to numbers of this form; we will discuss below how this can be achieved in the balanced flow computation. We can either choose $\varepsilon$ as a sufficiently small numerical value or, more elegantly, treat $\varepsilon$ as a symbolic positive infinitesimal. 

\begin{lemma}\label{lexicography}  Let $A$, $B$, $C$, and $D$ be positive rationals. Then 
$A B^\varepsilon > C D^\varepsilon$ for all sufficiently small $\varepsilon$ iff $A > C$ or $A = C$ and $B > D$. Also, $A B^\varepsilon / C D^\varepsilon = (A/C) \cdot (B/D)^\varepsilon$. 
\end{lemma}
\begin{proof} Obvious. \end{proof}

So prices and utilities are pairs $(A,B)$ with $A$ and $B$ positive rationals. The pair $(A,B)$ stands for $AB^\varepsilon$ and comparisons are lexicographic.

\subsection{The Implementation of a Phase: A Review of~\cite[Section 3]{\DM}.}\label{ssec: Implementation of a Phase}

Duan and Mehlhorn~\cite{\DM} already used perturbation in their algorithm, but for the purpose of keeping the cost of arithmetic low and not for the purpose of removing degeneracies. Fortunately, their perturbation extends to what is needed here. We review their use of perturbation adapted to our modified definition of $S$. 

Let $L = 128 n^{5n + 5} U^{4n}$. A real number $b$ is an \emph{additive $1/L$-approximation or simply additive approximation} of a real number $a$ if $\abs{a - b} \le 1/L$. It is a \emph{multiplicative $(1 + 1/L)$-approximation or simply multiplicative approximation} if $a/b \in [1/(1 + 1/L),1 + 1/L]$. 
\cite{\DM} approximates utilities by powers of $(1 + 1/L)$. For a utility $u_{ij} \in [1..U]$, let $e_{ij} \in \N$ be such that $\tu_{ij} \define (1 + 1/L)^{e_{ij}}$ is a multiplicative approximation of $u_{ij}$. Moreover, for part A of the algorithm, \cite{\DM} restricts prices to the form $(1 + 1/L)^k$, $0 \le k \le K$, where $K$ is chosen such that $(nU)^n \le (1 + 1/L)^K$. This choice of $K$ guarantees that the full range $[1..(nU)^n]$ of potential prices is covered. Then $K = O(n L \log(nU))$. They represent a price $p_i = (1 + 1/K)^{e_i}$ by its exponent $e_i \in \N$. Here, $\N$ denotes the natural numbers including zero. The bitlength of $e_i$ is $\log K = O(n \log (nU))$. 

The surplus vector with respect to $p$ is computed only approximately. To this end, they replace each price $p_i$ by an approximation $\hp_i$ with denominator $L$ (note that the denominator of $p_i$ might be as large as $L^K$) and compute a balanced flow in a modified equality network $N(p,\hp)$. The approximation $\hp_i$ is a rational number with denominator $L$ and an additive and a multiplicative approximation of $p_i$. 

\newcommand{\hq}{\hat{q}}

For two price vectors $p$ and $\hp$, the equality network $N(p,\hp)$ has its edge set determined by $p$ and its edge capacities determined by $\hp$, i.e., it has 
\begin{compactenum}[\hspace{\parindent}--]
    \item an edge $(s,b_i)$ with capacity $\hp_i$ for each $b_i \in B$,
    \item an edge $(c_i,t)$ with capacity $\hp_i$ for each $c_i \in C$, and 
    \item an edge $(b_i,c_j)$ with infinite capacity whenever $\tu_{ij}/p_j=\max_\ell \tu_{i\ell}/p_\ell$. We use $E_p$ to denote this set of edges.
\end{compactenum}\smallskip

Let $\hf$ be a balanced flow in $N(p,\hp)$. For a buyer $b_i$, let $\hr(b_i) = \hp_i - \hf_{si}$ be its surplus, for a good $c_i$, let $\hr(c_i) = \hp_i - \hf_{it}$ be its surplus. We define $S$ as in Section~\ref{sec: algorithm} except that we use the surplus vector $\hr$ instead of the surplus vector $r$. Recall that~\cite{\DM} uses simpler definition. 
Since prices are now restricted to powers of $1 + 1/L$, the update factor $x$ has to be a power of $1 + 1/L$, and hence they need to modify its definition and computation. They compute $x$ in two steps. They first compute a factor $\hx$ from $\hp$ as in Section~\ref{sec: algorithm} and then obtain $x$ from $\hx$ by rounding to a near power of $1 + 1/L$. They use $x$ to update the price vector $p$. The prices of all goods in $\Gamma(B(S))$ are multiplied by $x$. 

% \begin{figure}[t]
% \centerline{\framebox{\parbox{5.0in}{
% \begin{tabbing}
% 555\=555\=555\=555\=\kill
% \>\parbox{0.9\textwidth}{Set $R = 32 e^2$, $\varepsilon={1}/(4n^{4n}U^{3n})$, $\Delta=n^5/\varepsilon$, and
% $\xmax = 1 + 1/(Rn^3)$;}\\[0.3em]
% \>Set $\hp_i = p_i=1$ for all $i$ and let $\hf$ be a balanced flow in $N(p,\hp)$;\\[0.3em]
% \>{\bf Repeat}\\
% \>\>Sort the buyers by their surpluses $\hr(b)$ in decreasing order: $b_1,b_2,...,b_n$; \\[0.3em]
% \>\>\parbox{0.8\textwidth}{Find the smallest  $\ell$ for which $\hr(b_\ell)/\hr(b_{\ell+1})>1+1/n$, and\\ 
% \mbox{}\hfill let $\ell=n$ when there is no such $\ell$;}\\[0.3em]
% \>\>Let $S=\hr(b_\ell)$ and $B(S) = \sset{b_1,\ldots,b_\ell}$; \\[0.3em]
% \>\>Compute $\hx = \min(\xeq(S),\hx_{23}(S),\hx_{24}(S),\xmax)$;\\[0.3em]
% \>\>Compute $x$ as an $1 + 1/L$ approximation of $\hx$;\\[0.3em]
% \>\>Replace $p_i$ by $x p_i$ for $c_i \in \Gamma(B(S))$;\\[0.3em]
% \>\>Let $\hp$ be computed from $p$ according to (\ref{rounded prices});\\[0.3em]
% \>\>Update $\hf$ to a balanced maximum flow in $N(p,\hp)$; \\[0.3em]
% \> {\bf Exit from the loop if} $|\hr(B)|<\varepsilon$; \\[0.3em]
% \>Round $p$ to equilibrium prices by the procedure in Figure~\ref{fig: final rounding};
% \end{tabbing}
% }}}\caption{The polynomial time algorithm}\label{fig: algo poly}
% \end{figure}

They define $\hx$ as the minimum of $\xeq(S)$, $\hx_{23}(S)$, $\hx_{24}(S)$,  and $\xmax$. The definition of $\xeq(S)$ is in terms of the rounded utilities:
\begin{equation*}
\begin{split}
\xeq(S) =\min \{ &\frac{\tu_{ij}}{p_j}\cdot\frac{p_k}{\tu_{ik}} \; | \\
&b_i\in B(S), (b_i,c_j)\in E_p, c_k\notin\Gamma(B(S)) \}.
\end{split}
\end{equation*}
By definition, $\xeq(S)$ is a power of $1 + 1/L$. They redefine $\xmax$ as a power of $1 + 1/L$ such that 
$1 + 1/(Rn^3) \ge \xmax \ge (1 + 1/R(n^3))/(1 + 1/L)^2$. The quantities $\hx_{23}(S)$, $\hx_{24}(S)$, and $\hx_2(S)$ are defined in terms of the price vector $\hp$ and the surplus vector $\hr$, e.g., 
\begin{equation*}\begin{split}
\hx_{23}(S) = \min \{& \frac{\hp_i + \hp_j - \hr(b_j)}{\hp_i + \hp_j - \hr(b_i)} \; | \\
&\text{$b_i$ is type 2 and $b_j$ is type 3 buyer}\}
\end{split}\end{equation*}
Then $x_{23}(S)$, $x_{24}(S)$, $x_2(S)$ are multiplicative approximations of $\hx_{23}(S)$, $\hx_{24}(S)$, and $\hx_2(S)$ and $x = \min(\xeq,x_{23}(S),x_{24}(S),\xmax)$. Clearly, $x$ is a multiplicative $(1 + 1/L)$-approximation of $\hx$. Finally, $x$ is used to update the prices: $p'_i = x p_i$ for each good $c_i \in \Gamma(B(S))$ and 
$p'_i = p_i$ for any $c_i \not\in \Gamma(B(S))$. 

It was shown in~\cite{\DM} that the above can be implemented with arithmetic (additions, subtractions, multiplications, divisions) on integers whose bitlength is $O(n \log (nU))$. The computation of multiplicative and additive approximations takes $O(n \log (nU))$ operations. 
The following lemma summarizes the review:

\begin{theorem}\label{nondegenerate bitcomplexity} For nondegenerate inputs, the algorithm computes equilibrium prices and requires $O(n^6 \log^2 (nU))$ arithmetic operations (additions, subtractions, multiplications, divisions) on integers of bitlength $O(n \log (nU))$. Its time complexity is $O(n^7 \log^3 (nU))$. 
\end{theorem} 
\begin{proof} The number of phases is $O(n^4 \log (nU))$. At the beginning of a phase, we have the price vectors $p$ and $\hp$ and know a balanced flow $f$ in the network $N(p,\hp)$. We determine the set $S$, and the quantities $\xeq$, $\hx_{23}$, $\hx_{24}$, and $\hx_2$. This requires $O(n^2)$ arithmetic operations. We then compute $x_{23}$, $x_{24}$, and $x_2$. This takes $O(n \log (nU))$ arithmetic operations (\cite[Lemma 17]{\DM}). We then perform the price update ($n$ arithmetic operations), compute $\hp_{\mathit{new}}$ from the new price vector $p_{\mathit{new}}$. This takes $O(n \log(nU))$ arithmetic operations for each price (\cite[Lemma 16]{\DM}) for a total of $O(n^2 \log (nU))$ arithmetic operation. Finally, we compute a balanced flow in the network 
$N(p_{\mathit{new}},\hp_{\mathit{new}})$ ($n^2$ arithmetic operations). In summary, each phase requires 
$O(n^2 \log (nU))$ arithmetic operations on integers of bitlength $O(n \log(nU))$. Since the cost of arithmetic on a RAM with logarithmic word length has cost linear in the bitlength, the time bound follows. \end{proof}

\subsection{Adaption to Degeneracy Removal.}\label{ssec: Degeneracy Removal}

Recall that we have chosen a distinct prime $q_e \le Q = 8 n^2 \ln n$ for every edge $e \in B \times C$ with $u_e > 0$. Let $L' = 8n \ceil{Q}^{2n}$. 

For each edge $e$ with $u_e > 0$, let $\ell_e$ be such that $(1 + 1/L')^{\ell_e}$ is a $1 + 1/L'$-multiplicative approximation of $q_e$. Consider any cycle $D$ and assume w.l.o.g~$(\prod_{e \in D_0} q_e)/\prod_{e \in D_1} q_e > 1$.
Then the product is at least $Q^{2n}/(Q^{2n} - 1)$, and hence 
\begin{align*}
\frac{1}{\ln(1 + \frac{1}{L'})} & \ln \left( \frac{\prod_{e \in D_0} (1 + \frac{1}{L'})^{\ell_e}}{\prod_{e \in D_1} (1 + \frac{1}{L'})^{\ell_e}} \right)\\
&\ge \frac{1}{\ln(1 + \frac{1}{L'})} \ln \left(\frac{Q^{2n}}{Q^{2n} - 1} \cdot \frac{1}{(1 + \frac{1}{L'})^{2n}}\right)\\
&\ge \frac{1}{2}\frac{L'}{Q^{2n}} - 2n > 0,
\end{align*}
where the next to last inequality follows from $x/2 \le \ln(1 + x) \le x$ for $x \le 1$, and the last inequality follows
from the definition of $L'$.

We now combine the perturbations described in Sections~\ref{ssec: Our Perturbation}, \ref{ssec: Implementation of a Phase}, and the paragraph above. We replace a utility $u_{ij} \in [1..U]$ by
\[   \tu_{ij}(\varepsilon) = (1 + \frac{1}{L})^{e_{ij}} (1 + \frac{1}{L'})^{\ell_{ij}\varepsilon}, \]
where $e_{ij}$ is as in Section~\ref{ssec: Implementation of a Phase} and $\ell_{ij}$ is as in the preceding paragraph. For part A of the algorithm, prices are restricted to the form $(1 + 1/L)^a (1 + 1/L')^{b \varepsilon}$, where $a$ is a nonnegative integer and $b$ is an integer. For the computation of balanced flows, we proceed as in Section~\ref{ssec: Implementation of a Phase}. In particular, we replace each price $p_i$ by an approximation $\hp_i$ (for this approximation, we ignore the infinitesimal part of the prices, i.e., $A B^\varepsilon$ is replaced by $A$.) and then compute a balanced flow in the network $N(p,\hp)$. The results of Section~\ref{ssec: Implementation of a Phase} carry over. We only have to replace $U$ by $\max(U,O(n^2 \log n))$ in the time bounds.

\begin{theorem}\label{theorem bitcomplexity} The algorithm computes equilibrium prices with $O(n^6 \log^2 (nU))$ arithmetic operations on integers of bitlength $O(n \log (nU))$. Its time complexity is $O(n^7 \log^3 (nU))$. 
\end{theorem} 

\section{A Real Implementation}

Omar Darwish and Kurt Mehlhorn~\cite{Darwisch-Mehlhorn:AD-Implementation} have implemented the algorithm by Duan and Mehlhorn and parts of the improvements described in this paper. For random utilities, the running time seems to scale significantly better than $n^7$.

\newcommand{\htmladdnormallink}[2]{#1}
%\bibliographystyle{plain}
%\bibliography{ref,../references}

\section*{Appendix: Orlin's Perturbation}\label{Orlin's Perturbation}

In~\cite{Orlin10}, Orlin sketches a perturbation technique. He suggests to replace a nonzero utility $u_{ij}$ by\footnote{In the paper, he writes $\tu_{ij} = u_{ij} + \varepsilon^{n i} + \varepsilon^j$. In personal communication, he corrected this to  $\tu_{ij} = u_{ij} + \varepsilon^{n^2 i} + \varepsilon^j$. } $\tu_{ij} = u_{ij} + \varepsilon^{n^2 i} + \varepsilon^j$, where $\varepsilon$ is a positive infinitesimal. He shows that for the perturbed utilities the equality graph $E_p$ for any price vector $p$ is a forest. He does not discuss how to maintain the prices and he suggests that the perturbation can be implemented at no extra cost. 

Consider the following example of a Fisher market. We have two buyers $b_0$ and $b_1$ and two goods $c_0$ and $c_1$. The budget of $b_0$ is equal to 2 and the budget of $b_1$ is equal to one. The utilities are $u_{00} = u_{10} = 2$ and $u_{01} = u_{11} = 1$. In an equilibrium, the prices are $p_0 = 2$ and $p_1 = 1$, and the equality graph $E_p$ consists of all four edges. Money flows either on all four edges, or on edges $(0,0)$, $(0,1)$ and $(1,0)$  or on edges $(0,0)$ and $(1,1)$. In any case, there is flow on the edge $(0,0)$. 

The perturbation yields utilities $u_{00} = 2$, $u_{01} = 1 + \varepsilon$, $u_{10} = 2 + \varepsilon^4$, $u_{11} = 1 + \varepsilon^4 + \varepsilon$. Orlin writes: ``If $E$ is the equality graph prior to perturbations, and if $E'$ is the equality graph after perturbations, then $E'$ is the maximum weight forest of $E$ obtained by assigning edge $(i,j)$ a weight of $in^2 + j$.'' Thus $E'$ is $E_p$ minus the edge $(0,0)$. The equilibrium prices are $p_0 = 1$ and $p_1 = 2$, and this is nowhere close to the true solution. 

Orlin ignores that the prices after perturbation are rational functions in $\varepsilon$. Assume $p_1 = 1$. Then $(0,0)$ belongs to the equality graph if $2/p_0 \ge ( 1 + \varepsilon)/1$ or $p_0 \le 2/(1 + \varepsilon)$, and $(0,0)$ and $(0,1)$ belong to the equality graph if $p_0 = 2/(1 + \varepsilon)$. Edge $(1,0)$ belongs to the equality graph if $(2 + \varepsilon^4)/p_0 \ge (1 + \varepsilon^4 + \varepsilon)/1$ or $p_0 \le (2 + \varepsilon^4)/(1 + \varepsilon^4 + \varepsilon)$. Since 
$(2 + \varepsilon^4)/(1 + \varepsilon^4 + \varepsilon) < 2/(1 + \varepsilon)$ for $\varepsilon$ sufficiently small and positive, we have the following possibilities:\smallskip
\begin{compactenum}[1.]
\item $p_0 < (2 + \varepsilon^4)/(1 + \varepsilon^4 + \varepsilon)$: The equality graph consists of $(0,0)$ and $(1,0)$.
\item $p_0 = (2 + \varepsilon^4)/(1 + \varepsilon^4 + \varepsilon)$: The equality graph consists of $(0,0)$, $(1,0)$, and $(1,1)$. 
\item  $(2 + \varepsilon^4)/(1 + \varepsilon^4 + \varepsilon) < p_0 < 2/(1 + \varepsilon)$: The equality graph consists of $(0,0)$, and $(1,1)$. 
\item  $p_0 = 2/(1 + \varepsilon)$: The equality graph consists of $(0,0)$ and $(0,1)$ and $(1,1)$
\item $p_0 > 2/(1 + \varepsilon)$: The equality graph consists of $(0,1)$ and $(1,1)$. 
\end{compactenum}\smallskip
Only cases 2, 3, and 4 lead to equality graphs for which the original problem has an equilibrium solution. The example demonstrates that prices must be maintained as rational functions in $\varepsilon$ in Orlin's scheme. Alternatively, that a concrete sufficiently small value is chosen for $\varepsilon$. Both alternatives incur additional cost. 

The equilibrium solution corresponds to case 4. The prices are $p_1 = (3 + 3\varepsilon)/(3 + \varepsilon)$ and $p_0 = 2/(1 + \varepsilon) \cdot p_1$. The flow is $1$ on the edge $(1,1)$, $2$ on the edge $(0,0)$, and $2 - p_0 = p_1 - 1$ on the edge $(0,1)$.

\end{document}